\documentclass{article}

\usepackage{lipsum}
\usepackage{amsfonts,amsmath,amssymb,fullpage}
\usepackage{amsthm,hyperref}
\usepackage{graphicx}
\usepackage{epstopdf}
\usepackage[noend]{algpseudocode}
\usepackage{algorithm}

\newcommand{\cD}{\mathcal{D}}
\newcommand{\cL}{{\cal L}}
\newcommand{\cP}{\mathcal{P}}
\newcommand{\cR}{{\cal R}}
\newcommand{\cS}{\mathcal{S}}
\newcommand{\cA}{\mathcal{A}}
\newcommand{\cT}{{\cal T}}
\newcommand{\cN}{{\cal N}}
\newcommand{\R}{\mathbb R}

\newcommand{\F}{\mathbb F}

\newcommand{\eps}{\varepsilon}
\newcommand{\B}{\mathbf{B}}
\newcommand{\be}{{\bf e}}
\def\pdi{{\mathfrak m}}
\newcommand{\E}{\hbox{\bf E}}

\newtheorem{theorem}{Theorem}[section]

\newtheorem{lemma}[theorem]{Lemma}
\newtheorem{claim}[theorem]{Claim}
\newtheorem{corollary}[theorem]{Corollary}
\newtheorem{definition}[theorem]{Definition}

\newcommand{\Sec}[1]{\hyperref[sec:#1]{Section\,\ref*{sec:#1}}} 
\newcommand{\Obs}[1]{\hyperref[obs:#1]{Observation~\ref*{obs:#1}}} 
\newcommand{\Eqn}[1]{\hyperref[eq:#1]{Eq. (\ref*{eq:#1})}} 
\newcommand{\Fig}[1]{\hyperref[fig:#1]{Fig.\,\ref*{fig:#1}}} 
\newcommand{\Tab}[1]{\hyperref[tab:#1]{Table\,\ref*{tab:#1}}} 
\newcommand{\Thm}[1]{\hyperref[thm:#1]{Theorem\,\ref*{thm:#1}}} 
\newcommand{\Fact}[1]{\hyperref[fact:#1]{Fact\,\ref*{fact:#1}}} 
\newcommand{\Lem}[1]{\hyperref[lem:#1]{Lemma\,\ref*{lem:#1}}} 
\newcommand{\Prop}[1]{\hyperref[prop:#1]{Prop.~\ref*{prop:#1}}} 
\newcommand{\Cor}[1]{\hyperref[cor:#1]{Corollary~\ref*{cor:#1}}} 
\newcommand{\Conj}[1]{\hyperref[conj:#1]{Conjecture~\ref*{conj:#1}}} 
\newcommand{\Def}[1]{\hyperref[def:#1]{Definition~\ref*{def:#1}}} 
\newcommand{\Alg}[1]{\hyperref[alg:#1]{Algorithm~\ref*{alg:#1}}} 
\newcommand{\Pro}[1]{\hyperref[pro:#1]{Procedure~\ref*{pro:#1}}} 
\newcommand{\Ex}[1]{\hyperref[ex:#1]{Ex.~\ref*{ex:#1}}} 
\newcommand{\Clm}[1]{\hyperref[clm:#1]{Claim~\ref*{clm:#1}}} 
\newcommand{\Stp}[1]{\hyperref[step:#1]{Step~\ref*{step:#1}}}
\newcommand{\Ch}[1]{\hyperref[chap:#1]{Chapter~\ref*{chap:#1}}}
\ifpdf
  \DeclareGraphicsExtensions{.eps,.pdf,.png,.jpg}
\else
  \DeclareGraphicsExtensions{.eps}
\fi

\title{Erasure-Resilient Property Testing\thanks{A preliminary version of this work will appear in the Proceedings of the International Colloquium on Automata, Languages and Programming 2016~\cite{DRTV16}.}} 

\author{
Kashyap Dixit\thanks{Department of Computer Science and Engineering, The Pennsylvania State University. These three authors were supported in part by NSF award CCF-1320814, NSF CAREER award CCF-0845701, NSF award CCF-1422975, Pennsylvania State University College of Engineering Fellowship and Pennsylvania State University Graduate Fellowship.
    ({\tt kashyap@cse.psu.edu, sofya@cse.psu.edu, nzm154@psu.edu}).}
\and Sofya Raskhodnikova\footnotemark[2]
\and Abhradeep Thakurta\thanks{Previously at Yahoo Labs. {\tt guhathakurta.abhradeep@gmail.com}}
\and Nithin Varma\footnotemark[2]
}
\date{}

\usepackage{amsopn}

\begin{document}

\maketitle

\begin{abstract}
Property testers form an important class of sublinear algorithms. In the standard property testing model, an algorithm accesses the input function $f :\cD \mapsto \cR$ via an oracle.
With very few exceptions, all property testers studied in this model rely on the oracle to provide function values at all queried domain points. However, in many realistic
situations, the oracle may be unable to reveal the function values at some domain points due to privacy concerns, or when some of the values
get erased by mistake or by an adversary. The testers do not learn anything useful about the property by querying those {\em erased} points. Moreover, the knowledge of a tester may enable an adversary to
erase some of the values so as to increase the query complexity of the tester arbitrarily or, in some cases, make the tester entirely useless.

In this work, we initiate a study of property testers that are resilient to the presence of {\em adversarially erased} function values. An $\alpha$-erasure-resilient $\eps$-tester is given parameters $\alpha,\eps\in (0,1)$, along with oracle access to a function $f$ such that at most an $\alpha$ fraction of function values have been erased. The tester does not know whether a value is erased until it queries the corresponding domain point. The tester has to accept with high probability if there is a way to assign values to the erased points such that the resulting function satisfies the desired property $\cP$. It has to reject with high probability if, for every assignment of values to the erased points, the resulting function has to be changed in at least an $\eps$-fraction of the non-erased domain points to satisfy $\cP$.

We design erasure-resilient property testers for a large class of properties. For some properties, it is possible to obtain erasure-resilient testers by simply using standard testers as a black box. However, there are more challenging properties for which all known testers rely on querying a specific point. If this point is erased, all these testers break. We give efficient erasure-resilient testers for several important classes of such properties of functions including monotonicity, the Lipschitz property, and convexity. Finally, we show a separation between the standard and erasure-resilient testing. Specifically, we describe a property that can be $\eps$-tested with $O(1/\eps)$ queries in the standard model, whereas testing it in the erasure-resilient model requires number of queries polynomial in the input size.
\end{abstract}



\section{Introduction}
\label{sec:intro}
In this paper, we revisit the question of how sublinear-time algorithms access their input.
With very few exceptions, all algorithms studied in the literature on sublinear-time algorithms have {\em oracle} access to their input\footnote{Sublinear-time algorithms with various distributional assumptions on the positions of the input the algorithms access have been investigated, for example, in~\cite{GGR98,BBBY12,GoldreichR15}. There is also a line of work, initiated by~\cite{BatuFRSW13}, that studies sublinear algorithms that access distributions, as opposed to fixed datasets. In this work, we focus on fixed datasets.}.
However, in many applications, this assumption is unrealistic.
The oracle may be unable to reveal parts of the data due to privacy concerns, or when some of the values
get erased by mistake or by an adversary.
Motivated by these scenarios, we propose to study sublinear algorithms that work with partially erased data.

Formally, we view a dataset as a function over some discrete domain $\cD$, such as $[n]=\{1,\dots,n\}$ or $[n]^d$. For example, the classical problem of testing whether a list of $n$ numbers is sorted in nondecreasing order can be viewed as a problem of testing whether a function $f:[n]\to\mathbb{R}$ is monotone (nondecreasing). Given a parameter $\alpha\in(0,1)$, we say that a function is {\em $\alpha$-erased} if at most an $\alpha$ fraction of its domain points are marked as ``erased'' or protected (that is, an algorithm is denied access to these values).
An algorithm that takes an $\alpha$-erased function as its input does not know which values are erased until it queries the corresponding domain points. For each queried point $x$, the algorithm either learns $f(x)$ or, if $x$ is an erased point, gets back a special symbol $\perp$. We study algorithms that work in the presence of {\em adversarial erasures}. In other words, the query complexity of an algorithm is the number of queries it makes in the {\em worst case} over all $\alpha$-erased input functions.

In this work, we initiate a systematic study of property testers that are resilient to the presence of adversarial erasures. An $\alpha$-erasure-resilient $\eps$-tester is given parameters $\alpha,\eps\in (0,1)$, along with oracle access to an $\alpha$-erased function $f$. The tester has to accept with high probability if $f$ can be restored to a function on the whole domain that satisfies the desired property $\cP$ and reject with high probability if every restoration of $f$ is $\eps$-far from $\cP$ on the nonerased part of the domain. This generalizes the standard property testing model of Rubinfeld and Sudan~\cite{RubinfeldSudan96} and Goldreich, Goldwasser and Ron~\cite{GGR98} . 

\subparagraph*{Generic transformations} Our first goal is to understand which existing algorithms in the standard property testing model can be easily made erasure-resilient.
We show (in~\Sec{gen-trans}) how to obtain erasure-resilient testers for some properties by using standard testers for these properties as black box. Our transformations apply to testers that query uniformly and independently sampled points, with some additional restrictions. More specifically, our transformations work for uniform {\em proximity oblivious testers} (POTs)~\cite{GoldreichR11a} and uniform testers for {\em extendable properties}. As a result, we are able to obtain erasure-resilient testers for 
being a low-degree polynomial~\cite{RubinfeldSudan96}, monotonicity over general poset domains~\cite{FLNRRS02}, convexity of black and white images~\cite{BermanMR15b}, and Boolean functions over $[n]$ having $k$ runs of 0s and 1s. 

\subparagraph*{Erasure-resilient testers for more challenging properties}
One challenge in designing erasure-resilient testers by using existing algorithms in the standard model as a starting point is that many existing algorithms are more likely to query certain points in the domain. Therefore, if these points are erased, the algorithms break. Specifically, the optimal algorithms for testing whether a list of numbers is sorted (and there are at least three different algorithms for this problem~\cite{EKK+00,BGJRW12,CS13a}) have this feature. Moreover, it is known that an algorithm that makes uniformly random queries is far from optimal: it needs $\Theta(\sqrt{n})$ queries instead of $\Theta(\log n)$ for $n$-element lists \cite{EKK+00,Fis04}.

There is a number of well studied properties for which all known optimal algorithms heavily rely on querying specific points. Most prominent examples include monotonicity, the Lipschitz properties and, more generally, bounded-derivative properties of real-valued functions on $[n]$ and $[n]^d$, as well as convexity of real-valued functions on $[n]$. It is especially challenging to deal with real-valued functions in our model, because there are many possibilities for erased values. We give efficient erasure-resilient testers for all aforementioned properties of real-valued functions in Sections~\ref{sec:mono-line}-\ref{sec:conv}.

\subparagraph*{Relationships to other models} We explore the relationship of erasure-resilient testing with other testing models in~\Sec{app-oth-test}. We provide (in~\Sec{separation}) a separation between our erasure-resilient model and the standard model. Specifically, we prove the existence of a property that can be tested with $O(1/\eps)$ queries in the standard model, but requires polynomially many queries in the length of the input in the erasure-resilient model. This result builds on the ideas of Fischer and Fortnow~\cite{FF06} that separate tolerant testing, defined by Parnas, Ron and Rubinfeld~\cite{PRR06}, from standard testing.

A tolerant tester for a property $\cal P$, given two parameters $\eps_1,\eps_2\in(0,1)$, where $\eps_1<\eps_2$, is required to, with probability at least $2/3$, accept inputs that are $\eps_1$-close to ${\cal P}$ and reject inputs that are $\eps_2$-far from $\cP$. Intuitively, the relationship of our erasure-resilient model to tolerant testing is akin to the relationship between error-correcting codes that withstand erasures and error-correcting codes that withstand general errors.
As shown in~\cite{PRR06}, tolerant testing is equivalent to approximating the distance of a given input to the desired property.
In~\Sec{dist-appr}, we prove that the existence of tolerant testers implies the existence of erasure-resilient testers with related parameters. Using this implication and existing tolerant testers for sortedness~\cite{SaksS10}, monotonicity~\cite{FattalR10}, and convexity~\cite{FatR}, we get erasure-resilient testers for these properties as corollaries. However, we obtain erasure-resilient testers for these properties with much better parameters in the technical sections of this article. We conjecture that erasure-resilient testing can be separated from tolerant testing in the same strong sense as in our separation of standard testing from erasure-resilient testing.

\subsection{The Erasure-Resilient Testing Model}\label{sec:model}
We formalize our erasure-resilient model for the case of property testing. Erasure-resilient versions of other computational models, such as tolerant testing, can be defined analogously.

\begin{definition}[$\alpha$-erased function]
Let $\cD$ be a domain, $\cR$ be a range, and $\alpha\in(0,1)$. A function\footnote{Any object can be viewed as a function. E.g., an $n$-element array of real numbers can be viewed as a function $f:[n]\to\R$, an image---as a map from the plane to the set of colors, and a graph---as a map from the set of vertex pairs to $\{0,1\}$.} $f :\cD \mapsto \cR\cup\{\perp\}$ is {\em $\alpha$-erased} if  $f$ evaluates to $\perp$ on at most an $\alpha$ fraction of domain points. The points on which $f$ evaluates to $\perp$ are called {\em erased}. The set of remaining (nonerased) points is denoted by $\cN$.
\end{definition}
A function $f$ is $\eps$-far from a property (set) $\cP$ if it needs to be changed on at least an $\eps$ fraction of domain points to obtain a function in $\cP$.
A function $f':\cD\to \cR$ that differs from a function $f$ only on points erased in $f$ is called a {\em restoration} of~$f$.

\begin{definition}[Erasure-resilient tester]\label{def:erasure-resilient-tester}
An {\em $\alpha$-erasure-resilient $\eps$-tester} of property $\cP$ gets input parameters $\alpha,\eps\in(0,1)$ and oracle access to an $\alpha$-erased function $f:\cD\to \cR\cup\{\perp\}$. It outputs, with probability\footnote{In general, the error probability can be any $\delta\in(0,1)$. For simplicity, we formulate our model and the results with $\delta=1/3$. To get results for general $\delta$, by standard arguments, it is enough to multiply the complexity of an algorithm by $\log 1/\delta$.} at least~2/3,

\begin{itemize}
\item {\bf accept} if there is a restoration $f':\cD\to \cR$ of $f$ that satisfies $\cP$;
\item {\bf reject} if every restoration $f':\cD\to \cR$ of $f$
needs to be changed on at least an $\eps$ fraction of $\cN$, the nonerased portion of $f$'s domain,  to satisfy $\cP$
(that is, $f'$ is $\eps\cdot \frac{|\cN|}{|\cD|}$-far from $\cP$).
\end{itemize}
The tester has {\em 1-sided error} if the first item holds with probability 1.
\end{definition}

Let $f_{|\cN}$ denote the function $f$ restricted to the set $\cN$ of nonerased points. We show (in~\Sec{gen-trans}) that if property ${\cal P}$ is extendable, we can define a property $\cP_\cN$ such that the erasure-resilient tester is simply required to distinguish the case that $f_{|\cN}$ satisfies $\cP_\cN$ from the case that it is $\eps$-far from satisfying it. For example, if $\cP$ is monotonicity of functions on a partially-ordered domain $\cD$ then $\cP_\cN$ is monotonicity of functions on $\cN$. (Most of the properties we consider in this article, including monotonicity, Lipschitz properties and convexity, are extendable properties.) Note that, even for the case of extendable properties, our problem is different from the standard property testing problem because the tester does not know in advance which points are erased.

\subsection{Properties We Consider}\label{sec:prop-studied}
Next we define properties of real-valued functions considered in this article and summarize previous work on testing them.
Most properties of real-valued functions studied in the property testing framework are for functions over
the {\em line} domain $[n]$ and, more generally, the {\em hypergrid} domain $[n]^d$.

\begin{definition}[Hypergrid, line]\label{def:hypergrid}
Given $n,d \in \mathbb{N}$, the hypergrid of size $n$ and dimension $d$ is
the set $[n]^d$ associated with an order relation $\preceq$, such that
$x \preceq y$ for all $x,y \in [n]^d$ iff $x_i \le y_i$ for all $i \in [d]$, where $x_i$ (respectively $y_i$) denotes
the $i^{\text{th}}$ coordinate of $x$ (respectively, $y$).
The special case $[n]$ is called a {\em line}.
\end{definition}
\noindent We consider domains that are subsets of $[n]^d$ to be able to handle arbitrary erasures on $[n]^d$.
\subparagraph{Monotonicity}
Monotonicity of functions, first studied in the context of property testing in~\cite{GGLRS00}, is one of the most widely investigated properties in this model~\cite{EKK+00, DGLRRS99,LR01,FLNRRS02,AC04,Fis04,HK04,BRW05,PRR06,ACCL07,BGJRW12,BCG+10,BBM11,CS13a,CS13c,BlaisRY14,CDJS15}.
A function $f:\cD \mapsto \R$, defined on a
partially ordered domain $\cD$ with order $\preceq$, is monotone
if $x \preceq y$ implies $f(x) \le f(y)$ for all $x,y \in \cD$.
The query complexity of testing monotonicity of functions $f:[n]\mapsto \R$ is $\Theta(\log n/\eps)$~\cite{EKK+00,Fis04}; for functions $f:[n]^d\mapsto\R$, it is $\Theta(d\log n/\eps)$~\cite{CS13a,CS13c}, and for functions over arbitrary partially ordered domains $\cD$, it is $O(\sqrt{|\cD|/\eps})$ \cite{FLNRRS02}.

\subparagraph{Lipschitz properties} Lipschitz continuity is defined for functions between arbitrary metric spaces, but was specifically studied for real-valued functions on hypergrid domains~\cite{JhaR13, AJMS12, CS13a, DiJh+13,BlaisRY14,CDJS15} because of applications to privacy~\cite{JhaR13,DiJh+13}.
For $\cD \subseteq [n]^d$ and $c\in\R$,
a function $f:\cD\mapsto \R$ is $c$-Lipschitz  if $|f(x)-f(y)|\le c\cdot ||x-y||_1$ for all $x,y\in \cD$, where
$||x-y||_1$ is the $L_1$ distance between $x$ and $y$. More generally, $f$ is $(\alpha,\beta)$-Lipschitz, where $\alpha < \beta$, if $\alpha\cdot ||x-y||_1 \leq |f(x)-f(y)|\le \beta\cdot ||x-y||_1$
 for all $x,y\in [n]^d$.
All $(\alpha,\beta)$-Lipschitz properties can be tested with $O(d\log n/\eps)$ queries~\cite{CS13a}.

\subparagraph{Bounded derivative properties (BDPs)}
The class of BDPs, defined by Chakrabartyet al.~\cite{CDJS15}, is a natural generalization of monotonicity and the $(\alpha,\beta)$-Lipschitz properties.
An ordered set $\B$ of $2d$ functions
	$l_1, u_1, l_2, u_2, \ldots, l_d, u_d: [n-1] \mapsto \R\cup \{\pm\infty\}$
	is a \emph{bounding family} if for all $r \in [d]$ and $y \in [n-1]$, $l_r(y) < u_r(y)$.
	Let $\B$ be a bounding family of functions and let $\be_r$ be the unit vector along dimension $r$.
	The property $\cP(\B)$ of being \emph{$\B$-derivative bounded} is the set of functions
	$f:[n]^d \mapsto \R$ such that $l_r(x_r) \leq  f(x+\be_r) - f(x) \leq u_r(x_r)$ for all $r \in [d]$ and $x \in [n]^d$ with $x_r\neq n$, where
	$x_r$ is the $r^{\text{th}}$ coordinate of $x$.	
The class of BDPs includes monotonicity and the $c$-Lipschitz property.
The bounding family for monotonicity is obtained by setting $l_r(y) = 0$ and $u_r(y) = \infty$ for all $r \in [d]$, and for the $c$-Lipschitz property, by setting $l_r(y) = -c$ and $u_r(y) = c$ for all $r \in [d]$. In general, different bounding families allow a function to be monotone in one dimension, $c$-Lipschitz in another dimension and so on.
Chakrabarty et al.~\cite{CDJS15} showed that the complexity of testing BDPs of functions $f:[n]^d\mapsto \R$ is $\Theta(d\log n/\eps)$.
A bounding family $\B = \{l_1,u_1,\ldots,l_d,u_d\}$ defines a quasi-metric $$\pdi_\B(x,y) :=  \sum_{r:x_r > y_r} \sum_{t = y_r}^{x_r-1}\! u_r(t)  -  \sum_{r:x_r < y_r}\sum_{t = x_r}^{y_r-1}\!l_r(t)$$ over points $x,y \in [n]^d$.
In~\cite{CDJS15}, the authors observe that for $\cD=[n]^d$, a function
$f:\cD \mapsto \R \in \cP(\B)$, the bounded derivative property defined by $\B$, iff $~\forall x, y \in \cD$, $f(x) - f(y) \leq \pdi_\B(x,y)$.
We use this characterization as our definition of BDPs for functions over arbitrary $\cD\subseteq [n]^d$.

\subparagraph{Convexity of functions}
A function $f:\cD\mapsto \R$ is convex if $f(t{\bf x}+(1-t){\bf y})\le tf({\bf x})+(1-t)f({\bf y})$ for all ${\bf x},{\bf y} \in \cD$ and $t\in [0,1]$. If $\cD\subseteq[n]$, equivalently, $f$ is convex if $\frac{f(y)-f(x)}{y-x}\leq\frac{f(z)-f(y)}{z-y}$ for all $x<y<z$.
Parnas, Ron and Rubinfeld~\cite{PRR03} gave a convexity
tester for functions $f:[n]\mapsto \R$ with query complexity $O(\log n/\eps)$. Blais, Raskhodnikova and Yaroslavtsev~\cite{BlaisRY14} gave an $\Omega(\log n)$ bound for nonadaptive testers for this problem.

\subsection{Our Results}
We give efficient erasure-resilient testers for all properties discussed in Section~\ref{sec:prop-studied}. All our testers have optimal complexity for the case with no erasures and have an additional benefit of not relying too heavily on the value of the input function at any specific point.

\subparagraph{Monotonicity on the line} We start by giving (in~\Sec{mono-line}) an erasure-resilient  monotonicity tester on~$[n]$.
\begin{theorem}[Monotonicity tester on the line]\label{thm:line-tester}
There exists a one-sided error $\alpha$-erasure-resilient $\eps$-tester for monotonicity of real-valued functions on the line $[n]$ that works for all $\alpha,\eps \in (0,1),$ with query complexity
$O\left(\frac{\log n}{\eps(1-\alpha)}\right).$
\end{theorem}
Without erasure resilience, the complexity of testing monotonicity of functions $f:[n]\mapsto\R$ is $\Theta(\log n/\eps)$ \cite{EKK+00,Fis04}. Thus, the query complexity of our erasure-resilient tester has optimal dependence on the domain size and on~$\eps$.

The starting point of our algorithm is the tester for sortedness from~\cite{EKK+00}. This tester picks a random element of the input array and performs a binary search for that element. It rejects if the binary search does not lead to the right position. The first challenge is that the tester always queries the middle element of the array and is very likely to query other elements that are close to the root in the binary search tree. So, it will break if these elements are erased. To make it resilient to erasures, we randomize the binary tree with respect to which it performs the binary search. The second challenge is that the tester does not know which points are erased. To counteract that, our tester samples points from appropriate intervals until it encounters a nonerased point.

To analyze the tester, we bound the expected number of queries required to traverse a uniformly random search path in an arbitrary binary search tree built over the nonerased points in an $\alpha$-erased $n$-element array (\Clm{erasure-density}). This expectation depends only on the depth of the tree and $\alpha$. This is the most interesting part of our analysis and captures the intuition that a randomized binary search for a uniformly random search point is biased towards visiting intervals containing a larger fraction of nonerased points.

\subparagraph{BDPs on the hypergrid} In Sections~\ref{sec:BDP}-\ref{sec:app-bdp},
we generalize our monotonicity tester in two ways: (1) to work over general hypergrid domains, and (2) to apply to all BDPs. We achieve it by giving (1) a reduction from testing BDPs on the line to testing monotonicity on the line that applies to erasure-resilient testers and (2) an erasure-resilient version of the dimension reduction from~\cite{CDJS15}.
\begin{theorem}[BDP tester on the hypergrid]\label{thm:hyp-tester-BDP}
For every BDP $\cP$ of real-valued functions on the hypergrid $[n]^d$, there exists a one-sided error $\alpha$-erasure-resilient $\eps$-tester that works for all $\alpha,\eps \in (0,1),$ where $\alpha \le \eps/970d$, with query complexity $O\left(\frac{d\log n}{\eps(1-\alpha)}\right)$.
\end{theorem}

Every known tester of a BDP for real-valued functions over hypergrid domains work by sampling an {\em axis-parallel line} uniformly at random and checking for violations on the sampled line. Our erasure-resilient testers also follow this paradigm. To check for violations on the sampled line, we use one iteration of our BDP tester for the line. We show (in~\Sec{limitation-example-hypercube}) the existence of $\alpha$-erased functions $f:\{0,1\}^d \mapsto \R$ that are $\eps$-far from monotone for $\alpha = \Theta(\eps/\sqrt{d})$ but do not have violations to monotonicity along any of the axis parallel lines (which are the edges of the hypercube, in this case). It implies that every tester for monotonicity that follows the paradigm above will fail when $\alpha=\Omega(\eps/\sqrt{d})$. Thus, some restriction on $\alpha$ in terms of $d$ and $\eps$ is necessary for such testers.

\subparagraph{Convexity on the line} Finally, in~\Sec{conv}, we develop additional techniques to design a tester for convexity (which is not a BDP) on the line.
The query complexity of our tester has the same dependence on $n$ and $\eps$ as in the standard convexity tester of Parnas et al.~\cite{PRR03}. The dependence on $n$ is optimal for nonadaptive testers \cite{BlaisRY14}, and the tester from~\cite{PRR03} is conjectured to be optimal in the standard model.
\begin{theorem}[Convexity tester on the line]\label{thm:conv-tester}
There exists a one-sided error $\alpha$-erasure-resilient $\eps$-tester for convexity of real-valued functions on the line $[n]$ that works for all $\alpha,\eps \in (0,1),$ with query complexity
$O\left(\frac{\log n}{\eps(1 - \alpha)}\right).$
\end{theorem}

Our algorithm for testing convexity combines ideas on testing convexity from~\cite{PRR03}, testing sortedness from~\cite{EKK+00}, and our idea of randomizing the search. The tester of~\cite{PRR03} traverses a uniformly random path in a binary tree on the array $[n]$ by selecting one of the half-intervals of an interval uniformly at random at each step. Instead of doing this, our tester samples a uniformly random nonerased search point and traverses the path to that point in a uniformly random binary search tree just as in our modification of the tester of~\cite{EKK+00}. This is done to bias our algorithm to traverse paths containing intervals that have a larger fraction of nonerased points. However, instead of checking whether the selected point can be found, as in our monotonicity tester, the convexity tester checks a more complicated ``goodness condition'' in each visited interval of the binary search tree. It boils down to checking that the slope of the functions between pairs of carefully selected points satisfies the convexity condition. In addition to spending queries on erased points due to sampling, like in the monotonicity tester, our tester also performs ``walking queries'' to find the nearest nonerased points
to the left and to the right of
the pivots in our random binary search tree. We show that the overhead in the query complexity due to querying erased points is at most a factor of $O(1/(1-\alpha))$.

\section{Generic transformations}\label{sec:gen-trans}
In this section, we explain our transformations that can make two classes of testers erasure-resilient: (1) uniform proximity oblivious testers (POTs) defined by Goldreich and Ron~\cite{GoldreichR11a} (\Thm{pot-transform-main}), and (2) uniform testers for extendable properties (\Thm{ext-transform-main}). 

\subsection{Uniform POTs}
POTs were studied by Goldreich and Ron in~\cite{GoldreichR11a}, Goldreich and Kaufman~\cite{GK11} and Goldreich and Shinkar~\cite{GoldreichS16}. We first define POTs.
\begin{definition}[\cite{GoldreichS16}]
Let $\cP$ be a property, let $\rho : (0,1] \mapsto (0,1]$ be a monotone function and let $c \in (0,1]$ be a constant. A tester $T$ is  a {\em $(\rho,c)$-POT} for $\cP$ if
\begin{itemize}
\item for every function $f \in \cP$, the probability that $T$ accepts is at least $c$, and
\item for every function $f \notin \cP$, the probability that $T$ accepts is at most $c - \rho(\eps_f)$, where $\eps_f$ denotes the
relative Hamming distance of $f$ to $\cP$.
\end{itemize}
\end{definition}
A POT that queries points sampled uniformly and independently at random from $\cD$ is called a uniform POT.
Next, we state our first generic transformation.
\begin{theorem}\label{thm:pot-transform-main}
If $T$ is a uniform ($\rho$,$c$)-POT for a property $\cP$ that makes $q$ queries, then there exists a uniform $\alpha$-erasure-resilient ($\rho'$,$c$)-POT $T'$ for $\cP$ that makes $q$ queries for all $\alpha < \rho(\eps_f\cdot (1-\alpha))/q$, where $\rho'(x) = \rho(x\cdot(1-\alpha)) - \alpha\cdot q$ for $x \in (0,1]$.

\end{theorem}
\begin{proof}
Let $\cP$ be a property of functions over a domain $\cD$. The tester $T'$ queries $q$ uniform and independent points from $\cD$. It accepts if the sample has an erased point. Otherwise, it runs $T$ on the $q$ sampled nonerased points and accepts iff $T$ accepts.

Consider an $\alpha$-erased function $f \in\cP$ and a restoration $f^r\in \cP$. The tester $T$ accepts $f^r$ with probability at least $c$. If $T$ accepts $f^r$ on querying a sample $S \subseteq \cD$, then $T'$ also accepts $f$ on $S$. Thus, the probability that $T'$ accepts $f$ is at least $c$.

A tuple $W \in \cD^q$ is a {\em witness} for a function $g\notin\cP$, if $T$ rejects upon sampling $W$. Consider an $\alpha$-erased function $f$ that is $\eps$-far from $\cP$. Every restoration $f^r$ of $f$ is $\eps_f(1-\alpha)$-far from $\cP$. Since $T$ rejects $f^r$ with probability at least $1-c +\rho(\eps_f(1-\alpha))$, at least $(1-c +\rho(\eps_f(1-\alpha)))\cdot|\cD|^{q}$ tuples in $\cD^q$ are witnesses for $f^r$. Erasing one point can affect at most $q\cdot |\cD|^{q-1}$ witnesses. Thus, erasing an $\alpha$ fraction of points can affect at most $\alpha\cdot q\cdot|\cD|^{q}$ witnesses. At least $(1-c+\rho(\eps_f(1-\alpha))-\alpha\cdot q)\cdot|\cD|^{q}$ out of $|\cD|^q$ tuples are witnesses with no points erased. The probability that $T'$ samples such a tuple (and rejects $f$) is at least $1-c+\rho(\eps_f(1-\alpha))-\alpha\cdot q = 1 - c +\rho'(\eps_f)$. Hence, the probability that $T'$ accepts $f$
is at most $c-\rho'(\eps_f)$. This probability is nonnegative for all $\alpha < \rho(\eps_f(1-\alpha))/q$.
\end{proof}

\subparagraph*{Low degree polynomials} We apply~\Thm{pot-transform-main} to a POT designed by Rubinfeld and Sudan~\cite{RubinfeldSudan96} for the property of being a polynomial of degree at most $d$ over a finite field $\F$ and get an $\alpha$-erasure-resilient $\eps$-tester for this property. Consider a function $f: \mathbb{F} \mapsto \F$ that we would like to test for being a polynomial of degree at most $d$. The tester from~\cite{RubinfeldSudan96} selects $d+2$ points uniformly and independently at random from $\F$ and checks whether there is a polynomial of degree at most $d$ that fits all these points (by interpolation). It accepts if there is such a polynomial and rejects otherwise. Call this tester $T$. The authors of~\cite{RubinfeldSudan96} prove that $T$ rejects with probability at least $\eps$ if $f$ is $\eps$-far from being a degree-$d$ polynomial. Therefore, $T$ is a $(\rho,1)$-POT for this problem, where $\rho$ is the identity function. 
By~\Thm{pot-transform-main}, there exists
an $\alpha$-erasure-resilient $(\rho',1)$-POT, say $T'$, that makes $d+2$ queries, where $\rho'(x) = x(1-\alpha) - \alpha\cdot(d+2) $.
The probability that $T'$ rejects a function $f$ that is $\eps$-far from being a degree-$d$ polynomial is at least $\eps(1-\alpha) - \alpha\cdot(d+2)$.  The corollary follows.
\begin{corollary}
For $\alpha < \frac{\eps}{d+2+\eps}$, we can $\alpha$-erasure-resilient $\eps$-test for the property of being a degree-$d$ polynomial over a field $\F$ using $O\left(\frac{d+2}{\eps(1-\alpha) - \alpha\cdot (d+2)}\right)$ uniform queries.
\end{corollary}
\subsection{Uniform testers for extendable properties}
We now define extendable properties and present our transformation for uniform testers for such properties. Given $\cS \subseteq \cT$, the {\em extension} of a function $f:\cS \mapsto \cR$
to a domain $\cT$ is a function $g:\cT \mapsto \cR$ that agrees with $f$ on every point in $\cS$.

\begin{definition}[Extendable property]\label{def:ext-prop}
For a domain $\cD$ and all $\cS \subseteq \cD$, let $\cP_\cS$ denote a property of functions  over domain $\cS$. The property $\bigcup_{\cS \subseteq \cD}\cP_{\cS}$ is extendable if, for all $\cS,\cT:\cS\subseteq \cT \subseteq \cD$,
\begin{itemize}
\item for every function $f\in\cP_\cS$, there is an extension $f'\in\cP_\cT$, and
\item  for every function $f$ that is $\eps$-far from $\cP_\cS$, every function $f'\in\cP_\cT$ differs from $f$ on at least an $\eps$ fraction of points in $\cS$.
\end{itemize}
\end{definition}

The next lemma is used in the proof of our generic transformation.
\begin{lemma}\label{lem:ext-prop-fact}
Let $\bigcup_{\cS \subseteq \cD}\cP_{\cS}$ be an extendable property. Consider
an $\alpha$-erased function $f$ over domain $\cD$ and let $\cN \subseteq \cD$ be the set of nonerased points in it. If $f \in \cP_\cD$, then $f_{|\cN} \in \cP_\cN$. If $f$ is $\eps$-far from $\cP_\cD$, then $f_{|\cN}$ is $\eps$-far from $\cP_\cN$.
\end{lemma}
\begin{proof}
Suppose that $f\in\cP_\cD$. Assume for the sake of contradiction that $f_{|\cN}\notin\cP_{\cN}$. Therefore, no extension of $f_{|\cN}$ to the domain $\cD$ will satisfy $\cP_\cD$.
This contradicts our assumption that $f\in\cP_\cD$. Now, suppose that $f$ is $\eps$-far from $\cP_\cD$. Then, every restoration of $f$ needs to be changed in at least an $\eps$ fraction of
nonerased points to satisfy $\cP_\cD$. Assume for the sake of contradiction that the relative Hamming distance of $f_{|\cN}$ to $\cP_\cN$ is $\eps' < \eps$. Let $g$ be the function in $\cP_\cN$ closest
to $f_{|\cN}$. Let $g^e$ be an extension of $g$ to $\cD$ that satisfies $\cP_\cD$. Define an extension of $f_{|\cN}$ to $\cD$, say $f^e$ as follows. The function $f^e$ takes the same values as $f_{|\cN}$ on points in
$\cN$ and takes the same values as $g^e$ on the remaining points. Note that $f^e$ is a restoration of $f$ as well. Clearly, $f^e$ can be made to satisfy $\cP_\cD$ by changing an $\eps' < \eps$ fraction of points on $\cN$, which contradicts the assumption that $f$ is $\eps$-far
from $\cP_\cD$.
\end{proof}

Our generic transformation for uniform testers for extendable properties follows.
\begin{theorem}\label{thm:ext-transform-main}
Let $q(\cdot,\cdot)$ be a function that is nondecreasing in the first argument and nonincreasing in the second argument. Let $\bigcup_{\cS \subseteq \cD}\cP_{\cS}$ be an extendable property. Suppose $T$ is a uniform one-sided error $\eps$-tester for the property $\bigcup_{\cS \subseteq \cD}\cP_{\cS}$, such that $T$ makes $q(|\cS|,\eps)$ queries from $\cS$ to test for $\cP_\cS$, for every $\cS\subseteq \cD$. Assume also that for every $\cS \subseteq \cD$, the probability that $T$ tests $\cP_\cS$ correctly does not decrease when it makes more queries. Then, there is a uniform one-sided error $\alpha$-erasure-resilient $\eps$-tester for $\cP_{\cD}$ that makes $O\left(q(|\cD|,\eps)/(1-\alpha)\right)$ queries 
for all $\alpha \in [0,1)$.
\end{theorem}
\begin{proof}
Let $Q = 2q(|\cD|,\eps)/(1-\alpha)$. Consider the tester $T'$ that samples $Q$ points uniformly and independently at random from $\cD$. If there are fewer than $q(|\cD|,\eps)$ nonerased points in the sample, $T'$ accepts. Otherwise, it runs $T$ on the sampled nonerased points and accepts iff $T$ accepts.

The expected number of nonerased points in a uniform sample of size $Q$ from $\cD$ is at least $Q\cdot(1-\alpha) = 2q(|\cD|,\eps)$. By the Chernoff bound, the probability that $T'$ samples fewer than $q(|\cD|,\eps)$ nonerased points is at most $e^{-q(|\cD|,\eps)/4}$. 

Consider an $\alpha$-erased function $f$ over domain $\cD$. Let $\cN$ be the set of nonerased points.
If $f \in \cP$, then $f_{|\cN}\in\cP_\cN$ by~\Lem{ext-prop-fact}, and the tester $T'$ always accepts.
Assume now that $f$ is $\eps$-far from $\cP$. Then $f_{|\cN}$ is $\eps$-far from $\cP_\cN$ by~\Lem{ext-prop-fact}. Therefore $T$ rejects with probability at least $2/3$ on a sample of size at least $q(|\cN|,\eps)$. Thus, by a union bound, the probability that $T'$ accepts is at most $1/3 + e^{-q(|\cD|,\eps)/4}$. This probability can be brought below $1/3$ by repeating $T'$ a small constant number of times, whenever $q(|\cD|,\eps) \ge 8$. 
\end{proof}
In the following, we show a few applications of~\Thm{pot-transform-main}.
\subparagraph*{Convexity of Images} A black and white image, represented by a function $f:S \mapsto \{0,1\}$ for a subset $S$ of $[n]^2$, is convex if and only if for every pair of points $u,v \in S$ such that $f(u) = f(v) = 1$, every point $t \in S$ on the line joining $u$ and $v$  satisfy $f(t) =1$. Convexity is an extendable property.
Testing whether an image, represented by a function $f:[n]^2 \mapsto \{0,1\}$, is convex has been studied by Berman, Murzabulatov and Raskhodnikova~\cite{BermanMR15b}. The authors of~\cite{BermanMR15b} give a one-sided error uniform $\eps$-tester for this property that makes $O(1/\eps^{4/3})$ uniform queries. Their proofs go through even if the domain of $f$ is an arbitrary subset of $[n]^2$.
The corollary now follows by applying~\Thm{ext-transform-main} to the tester in~\cite{BermanMR15b}.
\begin{corollary}
There is an $\alpha$-erasure-resilient $\eps$-tester for convexity of black and white images that makes $O\left(\frac{1}{(1-\alpha)\eps^{4/3}}\right)$ uniform queries, where $\eps \in (0,1/2)$, $\alpha \in (0,1)$.
\end{corollary}

\subparagraph*{Monotonicity over poset domains} A real-valued function $f$ defined on a partially ordered domain is monotone if the function values respect the order relation of the poset. Monotonicity is an extendable property. The tester by Fischer et al.~\cite{FLNRRS02} samples $O(\sqrt{N/\eps})$ points uniformly at random and checks for violations to monotonicity among them. The corollary follows by applying~\Thm{ext-transform-main} to this tester.
\begin{corollary}
There is an $\alpha$-erasure-resilient uniform $\eps$-tester for monotonicity of real-valued functions over $N$ element posets that makes $O\left(\frac{1}{(1-\alpha)}\cdot\sqrt{\frac{N}{\eps}}\right)$ queries, where $\alpha \in (0,1)$.
\end{corollary}

\subparagraph*{Boolean functions with $k$-runs} 
A function $f : [n] \mapsto \{0,1\}$ has {\em $k$ runs} if the list $f(1),f(2),\dots, f(n)$  has at most $k-1$ alternations of values. The problem is to test whether a given function $f:[n] \mapsto \{0,1\}$ has $k$ runs or is $\eps$-far from this property. Kearns and Ron~\cite{KR00} studied a relaxation of this problem. Specifically, they showed that $O(1/\eps^2)$ queries suffice to test whether a Boolean function has $k$ runs or is $\eps$-far from being a $k/\eps$-run function. They also developed a uniform $O(\sqrt{k}/\eps^{2.5})$-query tester for this relaxation and proved that every uniform $\eps$-tester for the $k$-run property requires $\Omega(\sqrt{k})$ queries. Balcan et al.~\cite{BBBY12} obtained a $O(1/\eps^4)$-query tester for this property in the active testing model. They also developed a uniform $O(\sqrt{k}/\eps^6)$-query tester\footnote{Both~\cite{KR00} and~\cite{BBBY12} study Boolean functions over $[0,1]$. We note that their algorithms will also work for Boolean functions over $[n]$.}. We show the following. 
\begin{theorem}\label{thm:krun}
For $\eps > k^2/n$, we can $\eps$-test if a Boolean function over $[n]$ has at most $k$ runs using $O\left(\min\left\{\frac{k\cdot \log k}{\eps},\frac{\sqrt{k}}{\eps^6}\right\}\right)$ uniform and independent queries. 
\end{theorem}

\begin{algorithm}
\caption{\textsc{Tester for $k$-run Boolean functions} ($\varepsilon$,$k$,$f:\left[n\right] \mapsto \{0,1\}$)}
\label{alg:modal-tester}
\begin{algorithmic}[1]
\State Query the values at $\frac{3(k+1)\cdot\log (k+1)}{\eps}$ points uniformly and independently at random.
\State \textbf{Reject} if the values of $f$ at these points alternate $k$ or more times with respect to the ordering on the domain; \textbf{accept} otherwise.
\end{algorithmic}
\end{algorithm}

Our tester for being a $k$-run function is given in~\Alg{modal-tester}. It always accepts a function $f$ that has at most $k$ runs. The following lemma implies~\Thm{krun}.
\begin{lemma}
If $f$ is $\eps$-far from being a $k$-run function,~\Alg{modal-tester} rejects with probability at least $2/3$.
\end{lemma}
\begin{proof}
For $j \in [n]$ and $b \in \{0,1\}$, let $T_{b,j}$ denote the set consisting of the smallest $\lceil n\cdot\eps/(k+1) \rceil$ points in the set $\{x: j \le x \le n \text{ and } f(x) = b\}$, that is, the set of points between $j$ and $n$ where $f$ takes the value $b$. For a set $S \subseteq [n]$, let $\max(S)$ denote the largest element in $S$.
We will first describe a process to construct a few disjoint subsets of $[n]$ with some special properties.
\begin{itemize}
\item Let $S_1 = T_{b,1}$ such that $\max(T_{b,1}) < \max(T_{1-b,1})$.
\item For $i \ge 2$, the sets $S_i$ are defined as follows. Let the value that $f$ takes on the elements in $S_{i-1}$ be $b$ and let $j = \max(S_{i-1})$. Set $S_i = T_{1-b,j+1}$. Stop if $\max(S_i) = n$ or $S_i = \emptyset$.
\end{itemize}

The sets that this process constructs have the following properties. All $S_i$'s are subsets of $[n]$. Each point in $S_{i+1}$ is larger than every point in $S_{i}$ for all $i$. The function $f$ takes the same value on all points in $S_i$ for all $i$. The value of $f$ on points in $S_{i+1}$ is the complement of the value of $f$ on points in $S_{i}$ for  all $i$.

Next, we show that our process constructs sets $S_1,S_2,\ldots S_{k+1}$ each of size $\lceil n\cdot\eps/(k+1) \rceil$, if $f$ is $\eps$-far from satisfying the property. Let the process construct nonempty sets $S_1,S_2,\ldots S_t$. Assume for the sake of contradiction that $t \le k$. Let $S'_1 =\{x: 1 \le x \le \max(S_1)\}$. Let $S'_{i} = \{x: \max(S_{i-1}) < x \le \max(S_i)\}$ for all $1 < i \le t$. Note that for all $i \in [t]$, if $f$ takes the value $b$ on elements in $S_i$, then $f$ takes the value $1 -b$ on elements in $S'_{i} \setminus S_i$. We will describe a function $f'$ that has at most $k$ runs.
Set the values of $f'$ on each $x \in S'_1\setminus S_1$ to the value that $f$ takes on $S_1$. For each $1 < i \le t$, set the values of $f'$ on $S_i$ to the value of $f$ on $S'_{i} \setminus S_i$. On the rest of the points, $f'$ takes  the same value as $f$. We will now show that $f'$ has at most $k$ alternating intervals. The function $f'$ takes the same value on points in $S'_1 \cup S'_2$. Also, for each $1 < i \le t$, the function $f'$ is constant on $S'_i$. Thus, $f'$ has at most $t$ runs. Also, $f'$ differs from $f$ in at most $t \cdot \lceil n\cdot\eps/(k+1) \rceil \le k \cdot \lceil n\cdot\eps/(k+1) \rceil \le n\eps$ points, for $k < \sqrt{n\eps}$. This is a contradiction.

Using the fact that $k+1$ such subsets exist, we show that the tester will detect a violation with high probability. For a particular $i$, the probability that none of the points selected by the algorithm lie in $S_i$ is at most
$$\left(1 - \eps/(k+1)\right)^{3(k+1)\log (k+1)/\eps}\le 1/(k+1)^3.$$Therefore, by a union bound, the probability that there exists an $i$ such that none of the points selected by the algorithm lies in $S_i$ is at most $(k+1)^{-2}<1/3$ for $k\ge 1$.
\end{proof}

Since the property of being a $k$-run function is extendable, applying~\Thm{ext-transform-main} to~\Thm{krun} yields the following corollary.
\begin{corollary}
For $\eps > k^2/n$ and $\alpha \in (0,1)$, we can $\alpha$-erasure-resilient $\eps$-test if a Boolean function over $[n]$ has at most $k$ runs using $O\left(\frac{1}{1-\alpha}\cdot\min\left\{\frac{k\cdot \log k}{\eps},\frac{\sqrt{k}}{\eps^6}\right\}\right)$ uniform queries.
\end{corollary}

\section{Erasure-Resilient Monotonicity Tester for the Line}
\label{sec:mono-line}
In this section, we prove~\Thm{line-tester}. Recall that, for a function $f:[n] \mapsto \R \cup \{\perp\}$, the set of nonerased points (the ones that map to $\R$) is denoted by $\cN$. The function $f$ is monotone if $x < y$ implies $f(x) \le f(y)$ for all $x,y \in \cN$. The tester does not know $\cN$ in advance.

We present our tester in~\Alg{line-test}. It has oracle access to $f$ and takes $\alpha$ and $\eps$ as inputs. In each iteration, it performs a randomized binary search for a nonerased index sampled uniformly at random (u.a.r.) from $\cN$ and rejects if it finds violations to monotonicity.
In the description of our tester, we use $I[i,j]$ to denote the set of natural numbers from $i$ until and including $j$.
We alternatively refer to it as the interval from $i$ to $j$.

\begin{algorithm}
\caption{Erasure-Resilient Monotonicity Tester for the Line}
\label{alg:line-test}
\begin{algorithmic}[1]
\State \textbf{Set} $Q = \left\lceil \frac{60\log n}{\eps(1-\alpha)}\right\rceil$.
\State \textbf{Accept} at any point if the number of queries exceeds $Q$.
\Loop\label{step:iteration-mono-line} $2/\varepsilon$ times: \label{step:mon-er}
\State Sample points uniformly at random from $I[1,n]$ and query them until we get a point $s\in \cN$.
\State \textbf{Set} $\ell \leftarrow 1$, $r \leftarrow n$.
\While{$\ell \le r$}
\State  Sample points uniformly at random from $I[\ell,r]$ and query them until we get a point $m\in \cN$.
\If {$s < m$} \textbf{set} $r \gets m - 1$ and \textbf{Reject} if $f(s) > f(m)$.
\EndIf
\If {$s > m$} \textbf{set} $\ell \gets m + 1$ and \textbf{Reject} if $f(s) < f(m)$.
\EndIf
\If{$s = m$} Go to~\Stp{mon-er}. \Comment{Search completed.}
\EndIf
\EndWhile
\EndLoop
\State \textbf{Accept}.
\end{algorithmic}
\end{algorithm}
Every iteration of~\Alg{line-test} can be viewed as a traversal of a uniformly random search path in a uniformly random binary search tree defined on the set $\cN$ of nonerased points.
Given a binary search tree $T$ over $\cN$, we associate every node of $T$
with a unique sub-interval $I$ of $I[1,n]$ as follows. The root of $T$ is associated with $I[1,n]$.
Suppose the interval associated with a node $\Gamma$ in $T$ that contains $s \in \cN$ is $I[i,j]$.
Then the interval associated with the left child
of $\Gamma$ is $I[i,s-1]$ and the interval associated with the right child of $\Gamma$ is $I[s+1,j]$.
A search path is a path from the root to some node $\Gamma$ of $T$.

If $f$ is $\eps$-far from monotone, we prove that, with high probability, the tester finds a violation.
It is easy to prove this, using a generalization of an argument from~\cite{EKK+00}, for the case when \Alg{line-test} manages to complete all iterations of~\Stp{iteration-mono-line} before it runs out of queries.
The challenge is that the algorithm might get stuck pursuing long paths in a random search tree and waste many queries on erased points. To resolve the issue of many possible queries to erased points,
we prove an upper bound on the expected number of queries made while traversing a uniformly random search path in a binary search tree on $\cN$. We combine this with the fact that the expected depth of a random binary search tree is $O(\log n)$ to obtain the final bound on the probability that the algorithm exceeds its query budget.

\subsection{Analysis}
We analyze the tester in this section. The query complexity of the tester is clear from its description.
The main statement of~\Thm{line-tester} follows from \Lem{line-test-corr-1}, proved next.
\begin{lemma}\label{lem:line-test-corr-1}
\Alg{line-test} accepts if $f$ is monotone, and rejects with probability at least $2/3$ if $f$ is $\eps$-far from monotone.
\end{lemma}
\begin{proof}
The tester accepts whenever $f$ is monotone.
To prove the other part of the lemma, assume that $f$ is $\eps$-far from monotone. Let $A$ be the event that the tester
accepts $f$. Let $q$ denote the total number of queries made. We prove that $\Pr[A] \le 1/3$.
The event $A$ occurs if either $q>Q$ or the tester does not find a violation in any of the $2/\eps$ iterations of~\Stp{mon-er}. Thus, $\Pr\left[A\right]\le \Pr\left[A|q\le Q\right]+\Pr\left[q> Q\right].$

First we bound the probability that the tester does not find a violation in one iteration of~\Stp{mon-er}, conditioned on the event that $q \le Q$.
Consider an arbitrary binary search tree $T$
defined over points in $\cN$. A point $s \in \cN$ is called {\em searchable} with respect to $T$ if~\Alg{line-test} does not detect a violation to monotonicity while traversing
the search path to $s$ in $T$. Consider two indices $i,j \in \cN$, where $i < j$, both
searchable with respect to $T$. Let $a \in \cN$ be the pivot corresponding to the lowest common ancestor of the leaves
containing $i$ and $j$. Since $i$ and $j$ are both {\em searchable}, it must be the case that $f(i) < f(a)$ and $f(a) < f(j)$ and
hence, $f(i) < f(j)$. Thus, for every tree $T$, the function restricted to the domain points that are {\em searchable} with respect to $T$ is monotone.
Therefore, if $f$ is $\eps$-far from monotone, for every binary search tree $T$, at least an $\eps$-fraction of the
points in $\cN$ are not {\em searchable}. Thus, the tester detects a violation with probability $\eps$ in each iteration.
Consequently, $\Pr\left[A|q\le Q\right]\le (1-\eps)^{\frac{2}{\eps}} < 1/4. $

In the rest of the proof, we bound $\Pr[q > Q]$. We state and prove a claim that bounds the expected number of queries to traverse a search path, for every binary search tree. Recall that a search path in a search tree $T$ is a path from the root to some node in $T$.
Let $I$ be an interval associated with a node $v$ of $T$ and let $\alpha_I$ denote the fraction of erased points in $I$. The number of queries to be made to sample a nonerased point from $I$ with uniform sampling is a geometric random variable with expectation $1/(1-\alpha_I)$. We define the {\em query-weight} of node $v$ to be this expectation. The query-weight of a search path is the sum of query-weights of the nodes on the path (which is the expected number of queries that the algorithm makes while traversing that path).
\begin{claim}\label{clm:erasure-density}
Consider an arbitrary binary search tree $T$ on $\cN$ of height $h$. The expected query-weight of a uniformly random search path in $T$ is at most $h/(1-\alpha)$.
\end{claim}
\begin{proof}
There are exactly $|\cN|$ search paths in $T$. Let $S$ denote the sum of query-weights of all the search paths. The expected query-weight is equal to $S/|\cN|$.

Consider a node $v$ in $T$ associated with an interval $I$. There are $|I|(1-\alpha_I)$ nonerased points in $I$. The search paths from the root of $T$ to all these nonerased points pass through $v$, and hence, the query-weight of $v$ gets added to the query-weights of all of those paths. Therefore, the total contribution of $v$ towards $S$ is $|I|$, since the query-weight of $v$ is $1/(1-\alpha_I)$.
Note that the intervals associated with nodes at the same level of $T$ are disjoint from each other. Therefore, the total contribution to $S$ from all nodes on the same level of $T$ is at most $n$. Hence the value of $S$ is at most $n\cdot h$. Observe that this quantity is independent of the fraction of erasures $\alpha$. Therefore, the expected query-weight of a search path is at most $n\cdot h/|\cN|$, which is at most $h/(1-\alpha)$, since $|\cN| \ge n\cdot(1-\alpha)$.
\end{proof}

\noindent We will next see a fact on the expected depth of a uniformly random binary search tree and combine it with the above claim to prove the required bound on the expected query-weight of a uniformly random search path in a uniformly random binary search tree.
\begin{claim}[\cite{R03}]\label{clm:exp-rbst}
If $H_n$ is the random variable denoting the height of a random binary search tree on $n$ nodes, then $\E[H_n]\le 5\log n$.
\end{claim}
\begin{corollary}\label{cor:exp-quer-rbst}
The expected number of queries made by~\Alg{line-test} to traverse a uniformly random search path in a uniformly random binary search tree on $\cN$ is at most $5\log n/(1-\alpha)$.
\end{corollary}
By linearity of expectation, the expected number of queries made by the tester over all its iterations is at most $10\log n/(\eps\cdot(1-\alpha))$. Applying Markov's inequality to $q$, we can then see that $\Pr[q > Q] \le 1/6$.
Therefore, the probability of the tester not finding a violation is at most $1/3$. This completes the proof of the lemma.
\end{proof}

\section{Erasure-Resilient Monotonicity Testers for the Hypergrid}\label{sec:BDP}
In this section, we present our erasure-resilient tester for monotonicity over hypergrid domains and
prove the following theorem, which is a special case of~\Thm{hyp-tester-BDP}. We present the 
erasure-resilient testers for general BDPs in~\Sec{app-bdp}.
\begin{theorem}\label{thm:mon-hyp}
There exists a one-sided error $\alpha$-erasure-resilient $\eps$-tester for monotonicity of real-valued functions on the hypergrid $[n]^d$ that works for all $\alpha,\eps \in (0,1),$ where $\alpha \le \eps/250d$, with query complexity $O(\frac{d\log n}{\eps(1-\alpha)})$.
\end{theorem}

 Let $\cL$ denote the set of all {\em axis-parallel lines} in the hypergrid. Our monotonicity
tester, which is described in~\Alg{hyp-mon-tester}, samples an {\em axis-parallel line}
uniformly at random in each iteration and does a randomized binary search for a uniformly randomly sampled nonerased
point on that line. It rejects if and only if a violation to monotonicity is found within its query budget.
To analyze the tester, we first state two important properties of a uniformly random axis-parallel line
in~\Lem{exp-far-mon} and~\Lem{era-exp-mon}, which we jointly call the erasure-resilient dimension reduction.
 The statements and proofs of more general versions of these lemmas, applicable to all BDPs,
 are given in~\Sec{app-bdp}.

\begin{lemma}[Dimension reduction: distance]\label{lem:exp-far-mon}
Let $\eps_f$ be the relative Hamming distance of an $\alpha$-erased function $f:[n]^d \mapsto \R \cup \{\perp\}$ from monotonicity. Given an axis-parallel line $\ell\in \cL$, let $f_\ell :[n] \mapsto \R \cup \{\perp\}$ denote the restriction of $f$ to $\ell$ and  let $\eps_{\ell}$ denote the relative Hamming distance of $f_\ell$ from monotonicity. Then $\E_{\ell\sim \cL}[\eps_{\ell}]\ge (((1-\alpha)\cdot\eps_f)/4d)-\alpha.$
\end{lemma}

\begin{lemma}[Dimension reduction: fraction of erasures]\label{lem:era-exp-mon}
Consider an $\alpha$-erased function \mbox{$f:[n]^d \mapsto \R\cup\{\perp\}$}. Given $\ell\in \cL$,
let $\alpha_{\ell}$ denote the fraction of erased points in $\ell$. Then, for every $\eta \in(0,1)$,
we have, $\Pr_{\ell\sim\cL}[\alpha_\ell>\alpha/\eta]\le \eta.$
\end{lemma}
\begin{algorithm}
\caption{Erasure-Resilient Monotonicity Tester for $[n]^d$}
\label{alg:hyp-mon-tester}
\begin{algorithmic}[1]
\Require parameters $\eps \in (0,1), \alpha \in [0,\eps/250d]$; oracle access to $f:\left[n\right]^d \rightarrow \mathbb{R}$
\State \textbf{Set} $Q = \lceil\frac{1200d\cdot\log n}{\eps(1-\alpha)}\rceil$.

\Loop $~\frac{12d}{\eps(1-\alpha) - 4d\alpha}$ times:
\State  Sample a line $\ell \in \cL$ uniformly at random.
\State Sample and query points u.a.r. from $\ell$ and query them until we get a point $s \in \cN$.
\State Perform a randomized binary search for $s$ on $\ell$ as in~\Alg{line-test}.
\State \textbf{Reject} if any violation to monotonicity is found.

\EndLoop
\State \textbf{Accept} at any point if the number of queries exceed $Q$.
\end{algorithmic}
\end{algorithm}
The query complexity of the tester is evident from its description. We will now prove its correctness in the following lemma, which will then imply~\Thm{mon-hyp}.
\begin{lemma}\label{lem:correctness-hyp-mon}
\Alg{hyp-mon-tester} accepts if $f$ is monotone, and rejects with probability at least $2/3$ if $f$ is $\eps$-far from monotone.
\end{lemma}
\begin{proof}
The tester accepts if $f$ is monotone. So, assume that $f$ is $\eps$-far from being monotone.
Let $A$ denote the event that the tester does not find a violation to monotonicity in any of its iterations.
If $q$ denotes the total number of queries made by the tester, we have, $\Pr[A] \le \Pr[A | q\le Q] + \Pr[q >Q].$

Let $t$ denote the number of iterations of the tester. Let $A_i$ denote the event that the tester does not find a violation
in its $i$-th iteration.
For $\ell \in \cL$, let $f_\ell$ denote $f$ restricted to the line $\ell$. Let $\eps_\ell$ denote
the relative Hamming distance of $f_\ell$ from monotonicity.
We have, $\Pr[A_i | q \le Q] = \sum_{\ell \in \cL} (1 - \eps_\ell)\Pr[\ell] = 1 - \E_{\ell \sim \cL} [\eps_\ell]$. By~\Lem{exp-far-mon} and the fact that $\eps_f \ge \eps$, we have,
$\E_{\ell \sim \cL} [\eps_\ell] \ge \frac{(1-\alpha)\cdot\eps_f}{4d}-\alpha \ge \frac{(1-\alpha)\cdot\eps}{4d}-\alpha.$
Therefore,
$$\Pr[A | q \le Q] = \prod_{i = 1}^{t} \Pr[A_i | q \le Q] \le \left(1 - \frac{(1-\alpha)\cdot\eps-4d\alpha}{4d}\right)^t < \frac{1}{10}. $$

It now remains to bound $\Pr[q > Q]$. Let $\eta$ stand for $1/10t$. Let $\alpha_i$ denote the fraction of erasures in the line sampled
during iteration $i$ and let $q_i$ denote the
number of queries made by the algorithm during iteration $i$. Let $G$ denote the
(good) event that $\alpha_i \le \alpha/\eta$ for all iterations $i \in [t]$. By~\Cor{exp-quer-rbst}, $\E[q_i | G] \le 5\eta\cdot\log n/(\eta - \alpha)$, and
by the linearity of expectation, $\E[q|G] \le \log n/(2(\eta - \alpha)) \le 120d\log n/(\eps(1-\alpha))$, where the last inequality follows from our assumption that $\alpha \le \eps/250d$. Using Markov's inequality, $\Pr[q >Q | G] \le 1/10.$ Also, by combining~\Lem{era-exp-mon} with a union bound, we can see that $\Pr[\overline{G}] \le 1/10$. Therefore, $\Pr[q > Q] \le \Pr[q > Q | G] + \Pr[\overline{G}] \le 1/5.$
\end{proof}

\section{Erasure-Resilient BDP Testing}~\label{sec:app-bdp}
In this section, we discuss our erasure-resilient testers for all bounded derivative properties over hypergrid domains
and prove~\Thm{hyp-tester-BDP}.
First, we show in~\Lem{line-BDP} that testing for any BDP on $[n]$ reduces to testing
 monotonicity on $[n]$. Next, we prove~\Lem{exp-far} and~\Lem{era-exp} that reduces the problem of erasure-resilient testing of a BDP over
hypergrid domains to testing of the same property over the line.

\subsection{Erasure-Resilient BDP Tester for the Line}\label{sec:BDP-reduction}

In \Lem{line-BDP}, we show that (erasure-resilient) testing of bounded derivative properties (BDPs) on the line reduces to monotonicity testing on the line and prove~\Thm{bdp-line-tester}.
As noted in~\Sec{prop-studied}, BDPs comprise of a large class of properties that have been studied in the property testing literature.

Given a function $f:[n] \mapsto \R \cup \{\perp\}$, and a bounded derivative property $\cP$, we first define the notion of a violated pair in $f$ with respect to $\cP$.
\begin{definition}[Violated pair]\label{def:violpair}
Given a function $f :[n] \mapsto \R \cup \{\perp\}$ and bounding family $\B$ consisting of functions $l,u:[n-1]\mapsto \R$, two points $x,y \in \cN$ such that $x < y$
{\em violate} the property $\cP(\B)$ with respect to $f$ if $f(x) - f(y) > \pdi_\B(x,y) = -\sum_{t = x}^{y-1}l(t)$ or $f(y) - f(x) > \pdi_\B(y,x)= \sum_{t = x}^{y-1}u(t)$. The pairs $(x,y)$
and $(y,x)$ are called {\em violated}.
\end{definition}
Consider a bounded derivative property $\cP$ of functions defined over $[n]$ and associated
bounding functions $l,u:\left[n-1\right]\mapsto \R$.
The following claim states that,
we may assume w.l.o.g.\ that $l(i)= -u(i)$ for all $i \in [n-1]$.
We use it in the proof of~\Clm{bdpmonred}.

\begin{claim}\label{clm:posneg}
Consider a function $f:\left[n\right]\rightarrow \R\cup\{\perp\}$ and a bounding function family $\B$ over $[n]$ with $l,u:\left[n-1\right]\mapsto \R$.
Let $g:[n] \mapsto \R\cup\{\perp\}$ be a function that takes the value $f(i)+\sum_{j=i}^{n-1}\frac{l(j)+u(j)}{2}$ for each $i \in \cN$ and is erased on the remaining points. Let $\B'$ be a bounding function family over $[n]$ with $l',u':[n-1] \mapsto \R$ such that $u'(i)=-l'(i)=\frac{u(i)-l(i)}{2}$ for all $i \in [n-1]$. Then $x,y \in \cN$ violate $\cP(\B)$ with respect to $f$ iff $x,y$ violate $\cP(\B')$ with respect to $g$.
\end{claim}

\begin{proof}
Note that $(x,y) \in \cN$, where $x<y$, is not violated with respect to $f$ if and only if $\max\{f(x)-f(y)-\pdi_\B(x,y), f(y)-f(x)-\pdi_\B(y,x)\}\le 0$. We have
\begin{align*}
g(x)-g(y)-\pdi_{\B '}(x,y)&=f(x)-f(y)+\sum_{i=x}^{y-1}\frac{u(i)+l(i)}{2}-\sum_{i=x}^{y-1}\frac{u(i)-l(i)}{2}\\
&=f(x)-f(y)-\sum_{i=x}^{y-1}l(i)=f(x)-f(y)-\pdi_\B(x,y).
\end{align*}

Also,
\begin{align*}
g(y)-g(x)-\pdi_{\B'}(y,x)&= f(y)-f(x)-\sum_{i=x}^{y-1}\frac{u(i)+l(i)}{2}-\sum_{i=x}^{y-1}\frac{u(i)-l(i)}{2}\\
&=f(y)-f(x)-\sum_{i=x}^{y-1}u(i)=f(y)-f(x)-\pdi_\B(y,x).
\end{align*}
Thus, $\max\{g(x)-g(y)-\pdi_{\B'}(x,y), g(y)-g(x)-\pdi_{\B'}(y,x)\} = \max\{f(x)-f(y)-\pdi_\B(x,y), f(y)-f(x)-\pdi_\B(y,x)\}$. The claim follows.
\end{proof}

The following claim shows a reduction from testing BDPs over $\left[n\right]$ to testing monotonicity over $\left[n\right]$.

\begin{claim}
\label{clm:bdpmonred}
Consider an $\alpha$-erased function $f:\left[n\right]\mapsto \R\cup\{\perp\}$ and bounding functions $l,u:\left[n-1\right]\mapsto \R$
such that $-l(i)=u(i)=\gamma(i)$ for all $i\in \left[n-1\right]$. Let $\cP$ be the BDP defined by $l$ and $u$. Let $g,h:\left[n\right]\mapsto \R\cup\{\perp\}$ be two functions that
take the values $g(i)=f(i)-\sum_{r=i}^{n-1}\gamma(r)$ and $h(i)=-f(i)-\sum_{r=i}^{n-1}\gamma(r)$ for all $i \in \cN$ and are erased on the remaining points. Then, the following conditions hold:
\begin{enumerate}
\item $x,y \in \cN$ violate $\cP$ with respect to $f$ iff $x,y$ violate monotonicity with respect to either $g$ or $h$.
\item If $f$ is in $\cP$, then both $g$ and $h$ are both monotone.
\item If $f$ is $\eps$-far from $\cP$, then either $g$ or $h$ is at least $\eps/4$-far from monotonicity.
\end{enumerate}
\end{claim}

\begin{proof}
Consider a pair $(i,j) \in \cN \times \cN$ where $i<j$. We have,
\begin{align*}
g(i)-g(j) = f(i)-f(j)-\sum_{r=i}^{j-1}\gamma(r)&\text{~~~~~~~~~~~~and }& h(i)-h(j) = f(j)-f(i)-\sum_{r=i}^{j-1}\gamma(r).
\end{align*}

If $(i,j)$ is not violated with respect to $f$, we have $f(j)-f(i)-\sum_{r=i}^{j-1}\gamma(r)\le 0$ and $f(i)-f(j)-\sum_{r=i}^{j-1}\gamma(r)\le 0$.
Thus, $(i,j)$ satisfies the monotonicity property with respect to $g$ and $h$.
If $(i,j)$ is violated with respect to $f$, then either $f(j)-f(i)-\sum_{r=i}^{j-1}\gamma(r)> 0$
or $f(i)-f(j)-\sum_{r=i}^{j-1}\gamma(r)> 0$. That is, either $g$
or $h$ violates monotonicity.

Define the violation graph $G_f$ as follows. The vertex set corresponds to $\cN$.
There is an (undirected) edge between $i \in \cN$ and $j \in \cN$ iff the pair $(i,j)$ violates
the property $\cP$. By Lemma 2.5 in \cite{CDJS15}, the size of every maximal matching is at
 least $\eps\cdot |\cN|/2$. Consider a maximal matching
 $M$ in $G_f$. From the discussion above, every edge in $M$ violates monotonicity
 with respect to either $g$ or $h$. Therefore,
at least $\eps\cdot |\cN|/4$ edges are violated with respect to at least one of $g$ and $h$.
 Assume w.l.o.g.\ that at least $\eps\cdot |\cN|/4$ edges from $M$ are
 violated with respect to $h$. One has to change the function value of at least
 one endpoint of each edge to repair it.
Since $M$ is a matching in the violation graph $G_h$ as well, at least $\eps\cdot |\cN|/4$
 function values of $h$ have to change to make $h$ monotone. This means that
 $h$ is at least $\eps/4$-far from monotone.
\end{proof}

Therefore, in order to test the bounded derivative property $\cP$ on $f$  with proximity parameter $\eps$, one can test monotonicity on $g$ and $h$ with proximity parameter $\eps/4$ and error probability $1/6$ and accept iff both tests accept.

\begin{lemma}
\label{lem:line-BDP}
Let $Q_{\tt mon}(\alpha,\eps,n)$ denote the query complexity of $\alpha$-erasure-resilient $\eps$-testing of monotonicity of real-valued functions on the line. Then, for every BDP, $\alpha$-erasure-resilient $\eps$-testing of real-valued functions on the line has query complexity $O(Q_{\tt mon}(\alpha,\eps/4,n))$. The same statement holds for 1-sided error testing.
\end{lemma}

The following theorem is a direct consequence of~\Lem{line-BDP} and~\Thm{line-tester}.
\begin{theorem}[BDP tester on the line]\label{thm:bdp-line-tester}
For every BDP $\cP$, there exists a one-sided error $\alpha$-erasure-resilient $\eps$-tester for $\cP$ of real-valued functions on the line that works for all $\alpha,\eps \in (0,1),$ with query complexity
$O\left(\frac {1}{1-\alpha}\cdot \frac{\log n}{\eps}\right).$
\end{theorem}

\subsection{Erasure-Resilient Dimension Reduction}
\label{sec:dimension-reduction}

In this section, we prove two important properties of a uniformly random axis parallel line in the hypergrid $[n]^d$.
We do this in~\Lem{exp-far} and~\Lem{era-exp}, which we jointly call erasure-resilient dimension reduction.
We first introduce some notation.

Let $g$ be an $\alpha$-erased function on $\cD$, and $\cN \subseteq \cD$ be the set of nonerased points in $g$.
The Hamming distance of $g$ from $\cP$, denoted by $\text{dist}(g,\cP)$, is the least number of nonerased points on which
every restoration of $g$ needs to be changed to satisfy $\cP$. The relative Hamming distance between $g$ and $\cP$ is $\text{dist}(g,\cP)/|\cN|$.
We use $g_{|\cS}$ to denote the restriction of $g$ to a subset $\cS \subseteq \cD$. Note that all these definitions make sense even for functions with no
erasures in them.

Let $\cP$ be a bounded derivative
property of functions defined over $[n]^d$. Let $\cL$ denote the set of all {\em axis-parallel lines} in $[n]^d$.
Let $\cP^i$ denote the set of functions over $\cD$ with no violations to $\cP$ along dimension $i$ for all  $i \in [d]$.
Consider an $\alpha$-erased function $f:[n]^d \mapsto \R\cup\{\perp\}$. Let $\cN \subseteq [n]^d$ denote the set of nonerased points in $f$.
Let $\ell \in \cL$ be an axis-parallel line. Let $\cN_\ell$ denote the set of nonerased points on $\ell$ and $f_{\ell}$ denote the function $f$ restricted to $\ell$.

\Lem{exp-far} shows that the expected relative Hamming distance of $f_\ell$ from $\cP$ is roughly proportional to the relative Hamming distance of $f$ from $\cP$. 
First, we prove \Clm{dim-red} that we use in our proof of~\Lem{exp-far}.

\begin{claim}\label{clm:dim-red}
$$\frac{1}{4} \text{{\em dist}} (f,\cP)\le \sum_{i=1}^{d} \text{{\em dist}} (f,\cP^i)+\alpha\cdot d \cdot n^d.$$
\end{claim}
\begin{proof}
Let $g :[n]^d \mapsto \R$ be a function in $\cP$ such that $\text{dist}(g_{|\cN},f_{|\cN})$ is minimum.
We define $f_*:[n]^d\mapsto \R$, a restoration of $f$, such that $f_*(x) = f(x)$ for all $x \in \cN$ and
$f_*(x) = g(x)$ for all $x \notin \cN$.
Note that $g$ is the function closest to $f_*$ in $\cP$.

Also, for all $i \in [d]$, let $f_*^i:[n]^d \mapsto \R$ in $\cP^i$ be such that $\text{dist}(f_{*|\cN}^i,f_{|\cN})$ is minimum. Therefore, we have,
\begin{align*}
\frac{1}{4}\text{dist}(f,\cP)&\le\frac{1}{4}\text{dist}(f_*,\cP) \\
&\le \sum_{i=1}^{d}\text{dist}(f_*,\cP^i)& \text{ by dimension reduction from~\cite{CDJS15}}\\
&\le \sum_{i=1}^{d}\text{dist}(f_*,f_*^i) &\text{ because $f_*^i \in \cP^i$}\\
&\le \sum_{i=1}^{d}\text{dist}(f,\cP^i)+d\cdot\alpha\cdot n^d.
\end{align*}
The last inequality holds because, by triangle inequality, for all $i\in[d],$ $$\text{dist}(f_*,f_*^i)\leq\text{dist}(f,\cP^i)+\alpha\cdot n^d.$$
\end{proof}

We now use \Clm{dim-red} to prove the first part of our dimension reduction.

\begin{lemma}[Dimension reduction: distance]\label{lem:exp-far}
Let $\eps_f$ be the relative Hamming distance of $f$ from $\cP$. Given $\ell\in \cL$,  let $\eps_{\ell}$ denote the relative Hamming distance of $f_\ell$ from $\cP$. Then
$$\E_{\ell\sim \cL}[\eps_{\ell}]\ge \dfrac{(1-\alpha)\cdot\eps_f}{4d}-\alpha.$$
\end{lemma}
\begin{proof}
There are $d$ axis-parallel directions and, therefore, $dn^{d-1}$ axis-parallel lines in $[n]^d$. Thus, the probability of picking a specific axis parallel line $\ell$ uniformly at random is $1/dn^{d-1}$. Let $\cL_i$ denote the set of axis parallel lines along dimension $i$. 
\begin{align*}
\E_{\ell\sim\cL}[\eps_{\ell}] &= \sum_{\ell \in \cL}\eps_{\ell}\cdot\Pr(\ell)\\
&=\sum_{i=1}^{d}\sum_{\ell \in \cL_i} \eps_{\ell}\cdot\Pr(\ell)\\
&=\dfrac{1}{dn^{d-1}}\cdot\sum_{i=1}^{d}\sum_{\ell \in \cL_i} \dfrac{\text{dist}(f_\ell,\cP)}{|\cN_\ell|}\\
&\ge \dfrac{1}{dn^{d}}\cdot\sum_{i=1}^{d}\sum_{\ell \in \cL_i} \text{dist}(f_\ell,\cP) & \text{since } |\cN_\ell| \le n\\
&= \dfrac{1}{dn^{d}}\cdot\sum_{i=1}^{d} \text{dist}(f,\cP^i)\\
&\ge \dfrac{1}{dn^{d}}\cdot\left(\dfrac{\text{dist}(f,\cP)}{4}-\alpha d\cdot n^d\right) &\text{by \Clm{dim-red}}\\
&\ge\dfrac{1-\alpha}{4d}\cdot\eps_f-\alpha.
\end{align*}
\end{proof}

We conclude this section with the second part of our dimension reduction.

\begin{lemma}[Dimension reduction: fraction of erasures]\label{lem:era-exp}
Consider an $\alpha$-erased function \mbox{$f:[n]^d \mapsto \R\cup\{\perp\}$}. Given an axis-parallel line $\ell\in \cL$,
let $\alpha_{\ell}$ denote the fraction of erased points in $\ell$. Then, for every $\eta \in(0,1)$,
$$\Pr_{\ell\sim\cL}[\alpha_\ell>\alpha/\eta]\le \eta.$$
\end{lemma}

\begin{proof}
Note that a uniformly randomly sampled point in $[n]^d$ is erased with probability $\alpha$. We can sample a point uniformly at random by first sampling a line $\ell\in \cL$ uniformly at random and then sampling a point uniformly randomly on $\ell$,
which is erased with probability $\alpha_\ell$. Therefore we have
$$\alpha=\sum_{\ell\in \cL}\Pr[\ell]\cdot\alpha_{\ell}=\E_{\ell\sim \cL}[\alpha_{\ell}].$$
The claim then follows from Markov's inequality.
\end{proof}
\subsection{Erasure-Resilient BDP Testers for the Hypergrids}
We now present our erasure-resilient tester for an arbitrary BDP $\cP$ and complete the proof of~\Thm{hyp-tester-BDP}.
Let $\B = \{\ell_i,u_i : i \in [d]\}$ be a bounding family for $\cP$ and let $\cL_i$ denote
the set of axis-parallel lines along dimension $i$. Our tester is given in~\Alg{hyp-bdp-tester}.
The analysis of this tester is very similar to that of~\Alg{hyp-mon-tester} and is omitted.
\begin{algorithm}
\caption{Erasure-Resilient BDP Tester for $[n]^d$}
\label{alg:hyp-bdp-tester}
\begin{algorithmic}[1]
\Require parameters $\eps \in (0,1), \alpha \in [0,\eps/970d]$; oracle access to $f:\left[n\right]^d \rightarrow \mathbb{R}$
\State \textbf{Set} $Q = \left\lceil\dfrac{4800d\cdot\log n}{\eps(1-\alpha)}\right\rceil$.

\Loop $~\dfrac{48d}{\eps(1-\alpha) - 4d\alpha}$ times:
\medskip
\State Sample a line $\ell \in \cL$ uniformly at random.
\State Define $g$ and $h$ from $f_\ell$, $\ell_i$ and $u_i$ as in~\Clm{bdpmonred} if $\ell$ is sampled from $\cL_i$.
\State Sample points u.a.r. from $\ell$ and query them until we get a point $s \in \cN$.
\State Perform a randomized binary search for $s$ on $\ell$ as in~\Alg{line-test}.
\State \textbf{Reject} if any violation to monotonicity is found in either $g$ or $h$.

\EndLoop
\State \textbf{Accept} at any point if the number of queries exceed $Q$.
\end{algorithmic}
\end{algorithm}
\subsection{Limitations of Dimension Reduction based Erasure-Resilient Testers}\label{sec:limitation-example-hypercube}
In this section, we show that when the fraction of erasures is large enough, dimension reduction based testers that sample axis parallel lines uniformly at random and check for violations on them, are bound to fail.
More precisely we prove the following claim.
\begin{lemma}\label{lem:limit}
For all $\varepsilon\in (0,1/2]$, there exists an $\alpha$-erased function $f:\{0,1\}^d\mapsto \R \cup \{\perp\}$, where $\alpha=\Theta(\varepsilon/\sqrt{d})$, such that for large enough $d$, the function $f$ is $\varepsilon$-far from monotone and no axis-parallel edge in $\{0,1\}^d$ is violated in $f$.
\end{lemma}
\begin{proof}
For the ease of exposition, we prove this lemma for $\varepsilon=1/2$. We note that similar calculations could extend this proof to any
$\varepsilon\in (0,1/2]$. For $x \in \{0,1\}^d$, the function $f$ is defined as:
\begin{displaymath}
   f(x) = \left\{
     \begin{array}{lr}
       \perp &  \text{ if }||x||_0 = d/2\\
       1 &  \text{if }||x||_0 < d/2 \\
       0 &\text{ otherwise.}
     \end{array}
   \right.
\end{displaymath}
Note that the function $f$, when restricted to $\cN$, is $1/2$-far from being monotone. Also, no
axis-parallel edge is violated with respect to monotonicity. This completes the proof for the case when $\varepsilon=1/2$, since $\alpha = \Theta(1/\sqrt{d})$.

For general $\varepsilon$, we can define the set of erased points to be the points in $\{0,1\}^d$, such that their Hamming weight is $\beta\cdot d$, where $\beta=\beta(\varepsilon)<1/2$ is chosen so that $|S|/2^d=\varepsilon$, where $S$ is the set of all
points in $\{0,1\}^d$ with Hamming weight less than $\beta\cdot d$. As in the above case, we set all points with Hamming weight smaller than $\beta\cdot d$ to $1$ and the ones with Hamming weight larger than $\beta\cdot d$ to be $0$. Similar calculations help us prove that for large enough $d$,
fraction of erased points is $\Theta(\varepsilon/\sqrt{d})$. In this case, $f$ is $\varepsilon$-far from monotone, but no axis-parallel
edge is violated in $f$ with respect to monotonicity.
\end{proof}

\section{Erasure-Resilient Convexity Tester for the Line}\label{sec:conv}

In this section,
we prove~\Thm{conv-tester}.
Given an $\alpha$-erased function $f:[n] \mapsto \R \cup \{\perp\}$, let $\nu_i$ denote the $i$-th nonerased domain point in $[n]$. The derivative
of $f$ at a point $\nu_i \in \cN$, denoted by $\Delta f(\nu_i)$, is $\frac{f(\nu_{i+1}) - f(\nu_i)}{\nu_{i+1} - \nu_i}$, whenever $\nu_{i+1} \le n$.
The function $f$ is convex iff $\Delta f(\nu_i) \le \Delta f(\nu_{i+1})$ for all $i \in [|\cN| - 2]$.
Our tester builds upon the ideas in the convexity tester from~\cite{PRR03}.

A high level idea of the tester is as follows. Our tester (\Alg{conv-er-tester}) has several iterations. Every iteration of the tester can be thought of as a traversal of a uniformly random {\em search} path of a
uniformly random binary search tree on $\cN$, just as~\Alg{line-test}.
For each interval on such a path, we
check a set of conditions computed based on the values at
some nonerased points in the interval, called {\em anchor points}, and
two real numbers, called the left and right slopes. More specifically,
we verify that the function restricted to the sampled nonerased points in the interval is convex,
by comparing the slopes across consecutive points.
The algorithm accepts if all the intervals it sees pass these checks.
\begin{algorithm}
\caption{Erasure-Resilient Convexity Tester}
\label{alg:conv-er-tester}
\begin{algorithmic}[1]
\Require parameters $\eps,\alpha \in (0,1)$; oracle access to $f:\left[n\right] \mapsto \R \cup \{\perp\}$.
\State Set $Q =\lceil\frac{180\log n}{\eps(1-\alpha)}\rceil$.
\State {\bf Accept} at any point if the number of queries exceeds $Q$.
\Loop $~2/\eps$ times \label{step:loop-convexity}
\State Sample points in $I[1,n]$ u.a.r and query them until we get a point $s \in \cN$. 
\State \textsc{Test-Interval($I[1,n],\emptyset,-\infty,+\infty, s$)} and {\bf Reject} if it rejects.
\EndLoop
\State \textbf{Accept}.
\end{algorithmic}
\end{algorithm}
The main steps in the analysis of the tester follows that of the analysis of~\Alg{line-test}.
To analyze the tester, we first prove that, with high probability, the algorithm does not run out of its
budget of queries $Q$.
For this, we classify the queries that the tester
makes into two kinds and analyze them separately. The queries where the tester repeatedly samples and queries from an interval
until it finds a nonerased domain point are called {\em sampling queries.} The queries where
the tester keeps querying consecutive points, starting from a nonerased point, until it gets
the next nonerased point are called {\em walking queries.}
In the proof of~\Lem{conv-test-corr}, we show that the expected number of
walking queries is at most twice the number of the expected number of the sampling queries
and use~\Cor{exp-quer-rbst} to bound the
expected number of sampling queries.
In the second part of the analysis we prove that, conditioned on the aforementioned event happening, in every iteration, with probability at least $\eps$, the tester will detect a violation while testing on a function that is $\eps$-far from being convex. This part draws ideas from the proof of correctness of the tester in~\cite{PRR03}.

\begin{algorithm}
\floatname{algorithm}{Procedure}
\caption{\textsc{Test-Interval($I[i,j],\cA = \{ a_1,a_2,\dots, a_k\},m_\ell,m_r,s$)}}
\label{pro:conv-procedure}
\begin{algorithmic}[1]
\Require interval $I[i,j]$; a set of nonerased points $\cA$; left slope $m_\ell \in \R$; right slope $m_r \in \R$; search point $s \in \cN$.

\State Sample points u.a.r. from $I[i,j]$ and query them until we get a point $x \in \cN$. \label{step:pivot-conv}
\State Sequentially query points $x+1,x+2\dots$ until we get a nonerased point $y$.\label{step:walk-query}
\State Sequentially query points $x-1,x-2\dots$ until we get the nonerased point $z$.\label{step:walk-query-2}
\State Let $(a_1,a_2,\ldots,a_{k})$ denote the sorted list of points in the set $\cA \cup \{x,y,z\}$.\label{step:anchor-conv}
\State Let $m_i = (f(a_{i+1})-f(a_i))/(a_{i+1}-a_i)$ for all $i\in [k-1]$.
\State {\bf Reject} if  $m_\ell\le m_1\le m_2\le \dots\le m_{k-1}\le m_r$ is not true.
\State Let $\cA'_\ell$ and $\cA'_r$ be the sets of points in $\cA$ that are smaller and larger than $x$, respectively.
\If {$s<x$}
\State {\bf Reject} if \textsc{Test-Interval($I[i,z],\cA'_\ell,m_\ell,\Delta f(z),s$)} rejects.
\EndIf
\If {$s>x$}
\State {\bf Reject} if \textsc{Test-Interval($I[y, j],\cA'_r,\Delta f(x),m_r,s$)} rejects.
\EndIf
\State \textbf{Accept}.
\end{algorithmic}
\end{algorithm}

\begin{lemma}\label{lem:conv-test-corr}
\Alg{conv-er-tester} accepts if $f$ is convex, and rejects with probability at least $2/3$ if $f$ is $\eps$-far from convex.
\end{lemma}
\paragraph*{Proof} The tester accepts whenever $f$ is convex.
To prove the other part of the lemma, assume that $f$ is $\eps$-far from being convex. Let $A$ be the event that the tester
accepts $f$. Let $q$ denote the total number of queries made. We have, $\Pr\left[A\right]\le \Pr\left[A|q\le Q\right]+\Pr\left[q> Q\right].$

By~\Cor{exp-quer-rbst}, the expected number of sampling
queries made in one iteration of the tester is at most $5\log n/(1-\alpha)$.

We will now bound the expected number of walking queries.
Consider an interval $I$ with $\alpha_I$ fraction of erasures in it. A point in $I$
can get queried as part of the walking queries if either the first nonerased
point to its right or the first nonerased point to its left on the line $[n]$ gets
sampled as the pivot of $I$.  For a
nonerased point $i \in I$, let $w(i)$ denote the number of walking queries
to be made if the algorithm samples $i$ as the pivot. Therefore $\sum_{i \in \cN \cap I} w(i) \le 2|I|$, since every point in $I$ gets counted at most twice in this sum.
There are at least $|I|(1-\alpha_I)$ non erased points in $I$ and each of them
could be the pivot in $I$ with equal probability.
Hence, the expected number of walking
queries that~\Alg{conv-er-tester} makes in $I$ is at most $2/(1-\alpha_I)$.
This is at most twice the expected number of sampling queries that
the algorithm makes in $I$.

Therefore, by the linearity of expectation,
the expected number of walking queries made in one iteration of
the tester is at most $10\log n/(1-\alpha)$. Thus, the expected value
of the total number of queries made by the tester in one iteration
is at most $15\log n/(1-\alpha)$ and that over all iterations is at most
$30\log n/\eps(1-\alpha)$. Thus, by Markov's inequality,
$\Pr[q > Q] \le 1/6$.

Next, we bound $\Pr[A|q \le Q]$. We first define some notation for that.
Consider a search path traversed by the algorithm. Let $I[i,j]$ be an interval on the
path. Consider the execution of \textsc{Test-Interval} (\Pro{conv-procedure}) called
with $I[i,j]$ as the first argument. We call the nonerased point $x$ sampled in~\Stp{pivot-conv} its {\em pivot}, the set of points $\cA'$ in~\Stp{anchor-conv} its {\em anchor set}
and the values $m_\ell$ and $m_r$ as its {\em left} and {\em right slopes}, respectively. That is, given a binary search tree $\cT$, we associate each interval appearing in the tree with a pivot, an anchor set and two slopes.

Consider a binary search tree $\cT$ and a function $f:[n] \mapsto \R\cup\{\perp\}$.
Let $I[i,j]$ be an interval appearing in $\cT$ with anchor set $\cA=\{ a_1,a_2,\dots, a_k\}$
and slopes $m_\ell$ and $m_r$ such that $a_i\le a_{i+1}$ for all $i \in [k-1]$.  Let
$m_i = (f(a_{i+1})-f(a_i))/(a_{i+1}-a_i)$  $\forall i\in [k-1]$.
\begin{definition}[Good Interval, Bad Interval]\label{def:good-bad-conv}
 An interval $I[i,j]$ is {\em good} if
 $m_\ell\le m_1\le m_2\le \dots\le m_{k-1}\le m_r$. Otherwise, it is {\em bad}.
\end{definition}
\begin{definition}[Violator Interval]\label{def:viol-conv}
An interval $I[i,j]$ is a {\em violator} if it is {\em bad} and all its ancestor intervals in $\cT$
are {\em good}.
\end{definition}
\begin{definition}[Witness]\label{def:witness-conv}
A nonerased domain point is a {\em witness} with respect to $\cT$ if it belongs
to a violator interval in $\cT$.
\end{definition}

We prove that if $f$ is $\eps$-far from being convex, then,
for every binary search tree $\cT$,
the fraction of nonerased domain points that are witnesses is at least $\eps$.
We start by assuming that there is a tree in which the fraction of witnesses is less than
	$\eps$. We show that we can correct the function values only on the witnesses and get a convex function,
	which gives a contradiction.
	
\begin{claim}\label{clm:bst-conv-far}
	If $f$ is $\eps$-far from convex, then the fraction of witnesses in every binary search tree $\cT$ is more than $\eps$.
\end{claim}

\begin{proof}
Assume for the sake of contradiction that there is a binary search tree $\cT$
such that the fraction of witnesses with respect to $\cT$ is at most $\eps$.
In the following, we will construct a convex function $g:[n] \mapsto \R \cup \{\perp\}$ by changing
the values of $f$ only on witnesses with respect to $\cT$. Since
the fraction of witnesses is at most $\eps$, functions $f$ and $g$ will differ on at most an $\eps$ fraction of nonerased domain points, which results in a contradiction.
	
Consider a violator interval $I[i,j]$ in $\cT$. Since, by our assumption, the fraction of witnesses is at most $\eps$, the interval $I[i,j]$ cannot be the whole interval for otherwise, we have a contradiction immediately. Let the anchor set and slopes associated with the {\em parent interval} of $I[i,j]$ be $\cA=\{ a_1,a_2,\dots, a_k\}$ and $m_\ell$ and $m_r$, respectively such that $a_i \le a_{i+1}$ for all $i \in [k-1]$.
Assume that $I[i,j]$ is the right child of its parent. The case when $I[i,j]$ is the left child of its parent is similar. Let $\{a_u,a_{u+1},\ldots,a_k\}$
be the set of points common to $I[i,j]$ and $\cA$.
By definition, $a_u$ is the smallest
nonerased domain point in $I[i,j]$. Also, the left slope of $I[i,j]$ is $(f(a_u) - f(a_{u-1}))/(a_u - a_{u-1})$ and its right slope is equal to $m_r$.

Let $m_v = (f(a_{v+1}) - f(a_v))/(a_{v+1} - a_v)$ for all integers $v$ such that $v \in [u - 1,k)$.  We define $g$ as follows.
\begin{itemize}
\item For each $t \in \{a_u,a_{u+1},\ldots,a_k\}$, set $g(t) = f(t)$ .
\item For each integer $v \in [u,k)$ and $t \in \cN \cap (a_v,a_{v+1})$, set
\begin{align*}
g(t) = f(a_v) + m_v\cdot(t - a_v)
\end{align*}
\item For each $t \in \cN$ such that $t > a_k$, set
$$g(t) =  f(a_k) + m_{k-1}\cdot(t - a_k).$$
\end{itemize}
Since $I[i,j]$ is a violator, the parent interval of $I[i,j]$ is good, by definition. This implies that $m_{u-1} \le m_u \le \ldots \le m_k \le m_r$. Therefore, the derivatives of nonerased points in $I[i,j]$ are non-decreasing with respect to $g$, by virtue of our assignment.

To prove that $g$ is convex, we first show that every interval in $\cT$ is good with respect to $g$.
	\begin{enumerate}
		\item Consider an interval $I$ in $\cT$ that is good with respect to $f$. If $I$  has no ancestors or descendants that are violators, it remains good with
		respect to $g$ as well, since $g(t) = f(t)$ for all $t \in I[i,j]$.
		\item  Consider an interval $I$ that has a descendant $I'$ that is a violator. The definition of $g$ on points in $I'$ ensures that $g(t) = f(t)$ for every point $t$ common to the anchor sequence of $I$ and the
		interval $I'$. Thus, $I$ remains good with respect to $g$.
		
		\item Consider a node $I$ that is either a violator or has a violator ancestor $I'$.
		By definition, the parent of $I'$ is good with respect to $f$. Therefore, by the definition of $g$ on $I'$, we have $\Delta g(t-1) \le \Delta g(t)$ for all $t \in \cN$ such that $t \in I'$. Therefore, $I'$ is
		good with respect to $g$, and hence $I$ is also good with respect to $g$.
	\end{enumerate}
We proved that every interval in the tree $\cT$ is good with respect to $g$.
We now prove that $g$ is convex. Consider a point $\nu_t \in \cN$ such that $2 \le t \le |\cN|-1.$

This point occurs in $\cT$ either as a pivot in a non-leaf interval or as the sole nonerased domain
point in a leaf interval. In the former case, the condition $\Delta f(\nu_{t-1}) \le \Delta f(\nu_{t})$ is part of the goodness condition of the corresponding interval and is satisfied.
In the latter case, $\Delta f(\nu_{t-1})$ and $\Delta f(\nu_{t})$ are the left and right
slopes of the leaf and are compared as part of the goodness condition of the leaf. Thus,
$\Delta f(\nu_{t-1}) \le \Delta f(\nu_{t})$ for all $\nu_t \in \cN$ such that $2 \le t \le |\cN|-1$.
Thus, $g$ is convex.
\end{proof}

We conclude our analysis by bounding the probability that the tester does not find a violation.
Since the search point $s$ is chosen uniformly at random from the set of nonerased domain points,
the probability that it is a witness is at least $\eps$ and thus, the tester detects a violation to convexity with probability at least $\eps$ in every iteration.
Therefore, $\Pr[A|q\le Q]$ is at most $(1-\eps)^{\frac{2}{\eps}} <1/6.$
\qed

\section{Relations to Other Testing Models}\label{sec:app-oth-test}
In this section, we describe the relationships between erasure-resilient testing model and the other models of property testing. We first describe a property that is easy to test in the standard model, but is hard to test in the erasure-resilient model. This effectively separates the erasure-resilient testing model from the standard model. We discuss this result in~\Sec{separation}. Next, we study the connection of erasure-resilient testing to that of distance approximation algorithms and show that the existence of distance approximation algorithms for a property implies erasure-resilient testing algorithms for the same property. We describe it in~\Sec{dist-appr}.

\subsection{Separation Between Erasure-Resilient and Standard Testing}\label{sec:separation}
In this section we prove the following theorem that shows a separation between erasure-resilient testing and standard testing.
\begin{theorem}\label{thm:separation}
There exists a property $R$ such that $R$ can be $\eps$-tested in the standard model using $O(1/\eps)$ queries. However, there exists some $c > 0$
such that for all $\alpha=\Omega(\frac{\log \log \log n}{n})$ and $\eps \in (0,1)$, every $\alpha$-erasure-resilient $\eps$-tester for $R$ has to make
at least $n^c$ queries.
\end{theorem}

The property $R$ in \Thm{separation} is the property that was proposed by Fischer and Fortnow~\cite{FF06} to separate tolerant testing~\cite{PRR03} from standard testing.
The first part of the theorem is already proved in their paper. Our proof for the other part closely
follows the proof in~\cite{FF06} that separates tolerant testing and standard testing.
We first recall some definitions from~\cite{FF06}.
\begin{definition}[PCP witness~\cite{FF06}]
Given a promise problem and a Boolean input $v_1,\ldots,v_n$, a (one-sided) PCP witness for the problem is a set of Boolean functions
$f_1,\ldots,f_l$, where $l$ is polynomial in $n$, satisfying the following:
\begin{itemize}
\item The number of variables each of the functions depend on is independent of $n$. These variables might include variables from $v_1,\ldots,v_n$
as well as from a set of additional variables $w_1,\ldots,w_m$ such that $m$ is polynomial in $n$.
\item If the input is a \textsc{Yes} instance of the promise problem, then there is an assignment to $w$'s such that all
$f_i$'s are satisfied.
\item If the input is a \textsc{No} instance of the promise problem, then for all assignments of values to $w$'s,
at most half of the functions are satisfied.
\end{itemize}
\end{definition}
\begin{definition}[PCP of proximity~\cite{BGHSV06,FF06}]
A PCP of proximity is a PCP witness for an $\eps$-testing promise problem.
\end{definition}
The following lemma talks about the existence of PCPs of proximity for properties having polynomial-sized circuits.
\begin{lemma}[\cite{BGHSV06}]
If $P$ is a property of $v_1\ldots,v_n$ that is decidable by a circuit of size $k$, and $t < \log \log k/\log \log \log k$, then there exists a PCP of proximity for $P$ with
distance parameter $1/t$. Moreover, the number of additional variables and the number of functions in the PCP of proximity are both bounded by $k^2$, and each function depends
on $O(t)$ variables.
\end{lemma}
The following lemma says that there exists a property that is computable in polynomial time but is hard to test efficiently in the standard model.
\begin{lemma}[\cite{BHR05}]\label{lem:cnfhard}
There exists a property $U$ that is computable in polynomial time but any $\frac{1}{3}$-test of which requires at least $\Omega(n)$ queries, where $n$ is the input size.
\end{lemma}
Now we describe the property $R$ that is hard to test when there are adversarial erasures in the input. Let $p(n)$ be a polynomial bound on the
size of a circuit computing $U$. Let $t = \lfloor \log \log \log p(n)\rfloor$. Consider a bit string of length $m = n\cdot (p(n))^2$.
Label the first $(n - t) (p(n))^2$ bits by $v_{i,j}$ where $i \in [(n - t)(p(n))^2/n]$ and $j \in [n]$. Label the remaining bits by $w_{i,j}$ where $i \in [(p(n))^2]$ and $j \in [t]$.
The string is said to have the property $R$ if all of the
following conditions hold:
\begin{itemize}
\item For each $1 < i \le (n - t)(p(n))^2/n$ and $1 \le j \le n$, we have $v_{1,j} = v_{i,j}$.
\item $v_{1,1},\ldots,v_{1,n}$ satisfy the property $U$.
\item
For every $j \in [t]$, the sequence $w_{1,j},\ldots,w_{(p(n))^2,j}$ is an assignment satisfying the PCP
of proximity for the string $v_{1,1},\ldots,v_{1,n}$ for property $U$ with distance parameter $1/j$.
\end{itemize}
\begin{theorem}[\cite{FF06}]
Property $R$ can be $\eps$-tested in the standard property testing model using $O(1/\eps)$ queries.
\end{theorem}

We can prove the following theorem.
\begin{theorem}
There exists some $c > 0$
such that for all $\alpha= \Omega(\frac{\log \log \log n}{n})$ and $\eps \in [0,1/4]$, every $\alpha$-erasure-resilient $\eps$-tester for $R$ makes
at least $m^c$ queries, where $m = n(p(n))^2$ is the size of the input.
\end{theorem}
\begin{proof}
Assume for the sake of contradiction that there exists an $\alpha$-erasure-resilient $\eps$-tester for $R$ that makes fewer than $m^c$ queries for all constants $c > 0$. Let $n$ be such that $m = n(p(n))^2$ and let $c'$ be such that $m^{c'} = o(n)$.


Given an instance $v_1,v_2,\ldots v_n$ for which we need to test $U$, we can construct a partially erased string $I$ of length $m = n(p(n))^2$ as follows. Let $t = \lfloor \log \log \log p(n)\rfloor$.
Let $v_{i,1},v_{i,2},\ldots ,v_{i,n}$ be set to the string $v_1,v_2,\ldots ,v_n$ for all $ i \in [(n - t)(p(n))^2/n]$, where $v_{i,1},v_{i,2},\ldots ,v_{i,n}$ denote the $i$-th block of $n$ bits in $I$ from left. Let the remaining bits of $I$ be set to
the erased symbol $\perp$. A query to this new string can be simulated by at most one query to the string $v_1,v_2,\ldots ,v_n$.

If $v_1,v_2,\ldots,v_n$ satisfies $U$, then the new string is a \textsc{Yes} instance of $R$
for erasure-resilient testing problem by the definition of erasure-resilient property testing model.
If $v_1,v_2,\ldots,v_n$ is $\frac{1}{3}$-far from satisfying $U$, then the new string is $(1 - \frac{t}{n})\cdot\frac{1}{3}$-far
from $R$, which is at least $\frac{1}{4}$-far for large enough $n$. The fraction of erasures in the new string is $t/n$, which
is $O(\frac{\log \log \log n}{n})$. Therefore, an $\alpha$-erasure resilient $\frac{1}{4}$-tester for $R$ making $m^{c'}$ queries
for $\alpha = \Omega(\frac{\log \log \log n}{n})$
will yield a $\frac{1}{3}$-tester for $U$ that makes $o(n)$ queries. This is
a contradiction.
\end{proof}
\subsection{Connections to Distance Approximation Algorithms}\label{sec:dist-appr}
Here we discuss the relationship between tolerant testing, defined by Parnas et al.~\cite{PRR06}, and erasure-resilient testing.
We define the tolerant property testing model formally in the following.
\begin{definition}[\cite{PRR06}]\label{def:tolerant}
An algorithm is said to be an $(\eps_1,\eps_2)$-tolerant tester for a property $\cP$ if, when given oracle access to a function $f$, the algorithm (i) accepts with probability at least $2/3$ if $f$ is $\eps_1$-close to $\cP$ and (ii) rejects with probability at least $2/3$ if $f$ is $\eps_2$-far from $\cP$, where $0\le\eps_1<\eps_2\le1$. The algorithm is said to be fully tolerant if it works as above for all $\eps_1 < \eps_2$, which are given as the inputs.
\end{definition}
Tolerant testers are intimately connected to algorithms that approximate the distance of functions to properties, when given oracle access to the functions. For a property $\cP$ and a function $f$, we denote by $\eps_\cP(f)$ the relative Hamming distance of $f$ to $\cP$.
\begin{definition}[\cite{PRR06,FattalR10}]
Let $\cP$ be a property of functions over $\cD$. Let $\eta \ge 1$ and $\delta\in [0,1)$. An algorithm $A$ is said to be an $\eta$-distance approximation algorithm with additive error $\delta$ for $\cP$, if, given oracle access to a function $f$, the algorithm outputs, with probability at least $2/3$, a value $\hat{\eps}$ such that $\frac{1}{\eta}\cdot\eps_\cP(f) - \delta \le \hat{\eps} \le \eps_\cP(f)$. If $A$ works for all $\delta \in [0,1)$, we call it an $\eta$-distance approximation algorithm.
\end{definition}

The authors in~\cite{PRR06} prove that distance approximation algorithms for a property imply tolerant testers for the same property. They also show that the existence of fully tolerant testers for a property implies the existence of distance approximation algorithms for the same property. We will now prove that the existence of distance approximation algorithms for a property implies the existence of (weak) erasure-resilient testers for the same property. \begin{theorem}\label{thm:conv-dist-approx}
Let $A$ be an $\eta$-distance approximation algorithm with additive error $\delta$ for a property $\cP$ of functions of the form $f:\cD \mapsto \cR$. Then there exists an $\alpha$-erasure-resilient $\eps$-tester $A'$ that makes the same number of queries as $A$ and works for all $\eps,\alpha \in (0,1)$ satisfying $\alpha < \frac{\eps - \delta\cdot\eta}{\eps+\eta}$.
\end{theorem}
\begin{proof}
Fix an element $e \in \cR$. Consider the following algorithm $A'$. The algorithm $A'$, when given oracle access to an $\alpha$-erased function $g:\cD \mapsto \cR \cup \{\perp\}$, queries points from $g$ in the same way as $A$. Whenever it queries an erased point, it assumes that the value at that point is $e$. This way it computes the distance estimate $\hat{\eps}$ that $A$ would compute if all the erased points had values equal to $e$. If $\hat{\eps} \le \alpha$, the algorithm accepts. Otherwise, it rejects.

Let $g:\cD \mapsto \cR \cup \{\perp\}$ be an $\alpha$-erased function. Let $g^r : \cD \mapsto \cR$ be the restoration of $g$ in which all erased points are assigned the value $e$. We can think of $A'$ as outputting an approximation to $\eps_\cP(g^r)$. If $g$ satisfies $\cP$, then $\eps_\cP(g^r) \le \alpha$. Since $\hat{\eps} \le \eps_\cP(g^r)$ with probability at least $2/3$, the algorithm will accept with high probability. If $g$ is $\eps$-far from $\cP$, then every restoration of $g$ is $\eps(1-\alpha)$-far from $\cP$, and hence $\eps_\cP(g^r) \ge \eps(1-\alpha)$. Since $\hat{\eps} \ge \frac{\eps_\cP(g^r)}{\eta} - \delta \ge \frac{ \eps(1-\alpha)}{\eta} - \delta > \alpha$ with probability at least $2/3$, the algorithm will reject with high probability. Note that the last inequality in the above expression follows from the restriction on $\alpha$. The theorem follows.
\end{proof}
We now revisit the properties discussed in Section~\ref{sec:prop-studied} for which tolerant testers are known and apply~\Thm{conv-dist-approx} to those testers to get erasure-resilient testers. The parameters of these testers are much worse than what we obtained in previous sections, especially in terms of the restrictions on $\alpha$.
\begin{corollary}
Let $1 < \eta < 2$. There exists an $\alpha$-erasure-resilient $\eps$-tester for monotonicity of real-valued functions over $[n]$ with query complexity $O((\frac{1}{\eps(\eta -1)})^{O(\frac 1 {\eta-1})}\cdot \log^c n)$ (where $c$ is a large absolute constant) that works for all $\alpha,\eps \in (0,1)$ such that $\alpha < \frac{\eps}{\eps + \eta}$.
\end{corollary}
\begin{corollary}
Let $ \delta \in [0,1]$. There exists an  $\alpha$-erasure-resilient $\eps$-tester for monotonicity of real-valued functions over $[n]^d$ with query complexity $\tilde{O}\left(\frac{\log n}{\delta^4}\right)$ that works for all $\alpha,\eps \in (0,1)$ such that $\alpha < (\eps- 5\delta\cdot d^2\log n)/(\eps + 5d^2\log n)$.
\end{corollary}
\begin{corollary}
There exists an  $\alpha$-erasure-resilient $\eps$-tester for convexity of real-valued functions over $[n]$ with query complexity $\tilde{O}\left(\frac{\log n}{\eps}\right)$ that works for all $\alpha,\eps \in (0,1)$ such that $\alpha < \frac{\eps}{\eps + 25}$.
\end{corollary}

\section{Conclusions and Open Problems}\label{sec:conclusion}

In this paper, we initiate a study of property testing in the presence of adversarial erasures. We design efficient erasure-resilient testers for several important properties such as monotonicity, the Lipschitz properties and convexity over different domains. All our testers for properties of functions on the line domain work for an arbitrary fraction of erasures. All our testers have only a small additional overhead of $O(1/(1-\alpha))$ in their query complexity in comparison to the query complexity of the currently best, and, in some cases, optimal, standard testers for the same properties. We also show that not all properties are easy to test in the erasure-resilient testing model by proving the existence of a property that is easy to test in the standard model but hard to test in the erasure-resilient model even for a small fraction of erasures. We now list some open problems.
\begin{itemize}
\item We show that tolerant testing is at least as hard as erasure-resilient testing. Determining if tolerant testing is strictly harder than erasure-resilient testing is an interesting direction.
\item The fraction of erasures that our monotonicity tester for hypergrid domains ($[n]^d$) can tolerate decreases inversely with $d$. We also show that an inverse dependence on $\sqrt{d}$ is necessary for testers that work by sampling axis-parallel lines uniformly at random and then test for the property on them. It is an interesting combinatorial question to determine the exact tradeoff between the fraction of erasures and the fraction of axis parallel lines that are far from monotone.
\end{itemize}
 \section*{Acknowledgments}
We thank Jalaj Upadhyay for comments on a draft of this article.

\bibliographystyle{alpha}
\bibliography{references-for-erasure-resilient-testing}

\newcommand{\etalchar}[1]{$^{#1}$}
\begin{thebibliography}{BCGM12}

\bibitem[AC06]{AC04}
Nir Ailon and Bernard Chazelle.
\newblock Information theory in property testing and monotonicity testing in
  higher dimension.
\newblock {\em Inf. Comput.}, 204(11):1704--1717, 2006.

\bibitem[ACCL07]{ACCL07}
Nir Ailon, Bernard Chazelle, Seshadhri Comandur, and Ding Liu.
\newblock Estimating the distance to a monotone function.
\newblock {\em Random Struct. Algorithms}, 31(3):371--383, 2007.

\bibitem[AJMR16]{AJMS12}
Pranjal Awasthi, Madhav Jha, Marco Molinaro, and Sofya Raskhodnikova.
\newblock Testing {Lipschitz} functions on hypergrid domains.
\newblock {\em Algorithmica}, 74(3):1055--1081, 2016.

\bibitem[BBBY12]{BBBY12}
Maria{-}Florina Balcan, Eric Blais, Avrim Blum, and Liu Yang.
\newblock Active property testing.
\newblock In {\em 53rd Annual {IEEE} Symposium on Foundations of Computer
  Science, {FOCS} 2012, New Brunswick, NJ, USA, October 20-23, 2012}, pages
  21--30, 2012.

\bibitem[BBM12]{BBM11}
Eric Blais, Joshua Brody, and Kevin Matulef.
\newblock Property testing lower bounds via communication complexity.
\newblock {\em Computational Complexity}, 21(2):311--358, 2012.

\bibitem[BCGM12]{BCG+10}
Jop Bri{\"{e}}t, Sourav Chakraborty, David Garc{\'{\i}}a{-}Soriano, and Arie
  Matsliah.
\newblock Monotonicity testing and shortest-path routing on the cube.
\newblock {\em Combinatorica}, 32(1):35--53, 2012.

\bibitem[BFR{\etalchar{+}}13]{BatuFRSW13}
Tugkan Batu, Lance Fortnow, Ronitt Rubinfeld, Warren~D. Smith, and Patrick
  White.
\newblock Testing closeness of discrete distributions.
\newblock {\em J. {ACM}}, 60(1):4, 2013.

\bibitem[BGH{\etalchar{+}}06]{BGHSV06}
Eli Ben{-}Sasson, Oded Goldreich, Prahladh Harsha, Madhu Sudan, and Salil~P.
  Vadhan.
\newblock Robust {PCP}s of proximity, shorter {PCP}s, and applications to
  coding.
\newblock {\em {SIAM} J. Comput.}, 36(4):889--974, 2006.

\bibitem[BGJ{\etalchar{+}}12]{BGJRW12}
Arnab Bhattacharyya, Elena Grigorescu, Kyomin Jung, Sofya Raskhodnikova, and
  David~P. Woodruff.
\newblock Transitive-closure spanners.
\newblock {\em SIAM J. Comput.}, 41(6):1380--1425, 2012.

\bibitem[BHR05]{BHR05}
Eli Ben{-}Sasson, Prahladh Harsha, and Sofya Raskhodnikova.
\newblock Some {3CNF} properties are hard to test.
\newblock {\em {SIAM} J. Comput.}, 35(1):1--21, 2005.

\bibitem[BMR15]{BermanMR15b}
Piotr Berman, Meiram Murzabulatov, and Sofya Raskhodnikova.
\newblock {\em Testing Convexity of Figures Under the Uniform Distribution},
  2015.
\newblock To appear in {\it SoCG 2016}.

\bibitem[BRW05]{BRW05}
Tugkan Batu, Ronitt Rubinfeld, and Patrick White.
\newblock Fast approximate pcps for multidimensional bin-packing problems.
\newblock {\em Inf. Comput.}, 196(1):42--56, 2005.

\bibitem[BRY14]{BlaisRY14}
Eric Blais, Sofya Raskhodnikova, and Grigory Yaroslavtsev.
\newblock Lower bounds for testing properties of functions over hypergrid
  domains.
\newblock In {\em {IEEE} 29th Conference on Computational Complexity, {CCC}
  2014, Vancouver, BC, Canada, June 11-13, 2014}, pages 309--320, 2014.

\bibitem[CDJS15]{CDJS15}
Deeparnab Chakrabarty, Kashyap Dixit, Madhav Jha, and C.~Seshadhri.
\newblock Property testing on product distributions: Optimal testers for
  bounded derivative properties.
\newblock In {\em Proceedings of the Twenty-Sixth Annual {ACM-SIAM} Symposium
  on Discrete Algorithms, {SODA} 2015, San Diego, CA, USA, January 4-6, 2015},
  pages 1809--1828, 2015.

\bibitem[CS13]{CS13a}
Deeparnab Chakrabarty and C.~Seshadhri.
\newblock Optimal bounds for monotonicity and lipschitz testing over hypercubes
  and hypergrids.
\newblock In {\em Symposium on Theory of Computing Conference, STOC'13, Palo
  Alto, CA, USA, June 1-4, 2013}, pages 419--428, 2013.

\bibitem[CS14]{CS13c}
Deeparnab Chakrabarty and C.~Seshadhri.
\newblock An optimal lower bound for monotonicity testing over hypergrids.
\newblock {\em Theory of Computing}, 10:453--464, 2014.

\bibitem[DGL{\etalchar{+}}99]{DGLRRS99}
Yevgeniy Dodis, Oded Goldreich, Eric Lehman, Sofya Raskhodnikova, Dana Ron, and
  Alex Samorodnitsky.
\newblock Improved testing algorithms for monotonicity.
\newblock In {\em Randomization, Approximation, and Combinatorial Algorithms
  and Techniques, Third International Workshop on Randomization and
  Approximation Techniques in Computer Science, and Second International
  Workshop on Approximation Algorithms for Combinatorial Optimization Problems
  RANDOM-APPROX'99, Berkeley, CA, USA, August 8-11, 1999, Proceedings}, pages
  97--108, 1999.

\bibitem[DJRT13]{DiJh+13}
Kashyap Dixit, Madhav Jha, Sofya Raskhodnikova, and Abhradeep Thakurta.
\newblock Testing the lipschitz property over product distributions with
  applications to data privacy.
\newblock In {\em {TCC}}, pages 418--436, 2013.

\bibitem[DRTV16]{DRTV16}
Kashyap Dixit, Sofya Raskhodnikova, Abhradeep Thakurta, and Nithin Varma.
\newblock Erasure-resilient property testing.
\newblock In {\em To appear in {\it ICALP 2016}}, 2016.

\bibitem[EKK{\etalchar{+}}00]{EKK+00}
Funda Erg{\"{u}}n, Sampath Kannan, Ravi Kumar, Ronitt Rubinfeld, and Mahesh
  Viswanathan.
\newblock Spot-checkers.
\newblock {\em J. Comput. Syst. Sci.}, 60(3):717--751, 2000.

\bibitem[FF06]{FF06}
Eldar Fischer and Lance Fortnow.
\newblock Tolerant versus intolerant testing for boolean properties.
\newblock {\em Theory of Computing}, 2(9):173--183, 2006.

\bibitem[Fis04]{Fis04}
E.~Fischer.
\newblock On the strength of comparisons in property testing.
\newblock {\em Inform.\ and Comput.}, 189(1):107--116, 2004.

\bibitem[FLN{\etalchar{+}}02]{FLNRRS02}
Eldar Fischer, Eric Lehman, Ilan Newman, Sofya Raskhodnikova, Ronitt Rubinfeld,
  and Alex Samorodnitsky.
\newblock Monotonicity testing over general poset domains.
\newblock In {\em Proceedings of the thiry-fourth annual ACM symposium on
  Theory of computing}, STOC '02, pages 474--483, New York, NY, USA, 2002. ACM.

\bibitem[FR]{FatR}
Shahar Fattal and Dana Ron.
\newblock Approximating the distance to convexity.
\newblock Unpublished manuscript. Uploaded at
  http://www.eng.tau.ac.il/~danar/Public-pdf/app-conv.pdf.

\bibitem[FR10]{FattalR10}
Shahar Fattal and Dana Ron.
\newblock Approximating the distance to monotonicity in high dimensions.
\newblock {\em {ACM} Transactions on Algorithms}, 6(3), 2010.

\bibitem[GGL{\etalchar{+}}00]{GGLRS00}
O.~Goldreich, S.~Goldwasser, E.~Lehman, D.~Ron, and A.~Samorodnitsky.
\newblock Testing monotonicity.
\newblock {\em Combinatorica}, 20:301--337, 2000.

\bibitem[GGR98]{GGR98}
Oded Goldreich, Shafi Goldwasser, and Dana Ron.
\newblock Property testing and its connection to learning and approximation.
\newblock {\em J. {ACM}}, 45(4):653--750, 1998.

\bibitem[GK11]{GK11}
Oded Goldreich and Tali Kaufman.
\newblock Proximity oblivious testing and the role of invariances.
\newblock In {\em Approximation, Randomization, and Combinatorial Optimization.
  Algorithms and Techniques - 14th International Workshop, {APPROX} 2011, and
  15th International Workshop, {RANDOM} 2011, Princeton, NJ, USA, August 17-19,
  2011. Proceedings}, pages 579--592, 2011.

\bibitem[GR11]{GoldreichR11a}
Oded Goldreich and Dana Ron.
\newblock On proximity-oblivious testing.
\newblock {\em {SIAM} J. Comput.}, 40(2):534--566, 2011.

\bibitem[GR16]{GoldreichR15}
Oded Goldreich and Dana Ron.
\newblock On sample-based testers.
\newblock {\em ACM Trans. Comput. Theory}, 8(2):7:1--7:54, April 2016.

\bibitem[GS16]{GoldreichS16}
Oded Goldreich and Igor Shinkar.
\newblock Two-sided error proximity oblivious testing.
\newblock {\em Random Struct. Algorithms}, 48(2):341--383, 2016.

\bibitem[HK08]{HK04}
Shirley Halevy and Eyal Kushilevitz.
\newblock Testing monotonicity over graph products.
\newblock {\em Random Struct. Algorithms}, 33(1):44--67, 2008.

\bibitem[JR13]{JhaR13}
Madhav Jha and Sofya Raskhodnikova.
\newblock Testing and reconstruction of lipschitz functions with applications
  to data privacy.
\newblock {\em {SIAM} J. Comput.}, 42(2):700--731, 2013.

\bibitem[KR00]{KR00}
Michael~J. Kearns and Dana Ron.
\newblock Testing problems with sublearning sample complexity.
\newblock {\em J. Comput. Syst. Sci.}, 61(3):428--456, 2000.

\bibitem[LR01]{LR01}
E.~Lehman and D.~Ron.
\newblock On disjoint chains of subsets.
\newblock {\em J.\ Combin.\ Theory Ser.\ A}, 94(2):399--404, 2001.

\bibitem[PRR03]{PRR03}
Michal Parnas, Dana Ron, and Ronitt Rubinfeld.
\newblock On testing convexity and submodularity.
\newblock {\em {SIAM} J. Comput.}, 32(5):1158--1184, 2003.

\bibitem[PRR06]{PRR06}
M.~Parnas, D.~Ron, and R.~Rubinfeld.
\newblock Tolerant property testing and distance approximation.
\newblock {\em J.\ Comput.\ System Sci.}, 6(72):1012--1042, 2006.

\bibitem[Ree03]{R03}
Bruce~A. Reed.
\newblock The height of a random binary search tree.
\newblock {\em J. {ACM}}, 50(3):306--332, 2003.

\bibitem[RS96]{RubinfeldSudan96}
Ronitt Rubinfeld and Madhu Sudan.
\newblock Robust characterizations of polynomials with applications to program
  testing.
\newblock {\em {SIAM} J. Comput.}, 25(2):252--271, 1996.

\bibitem[SS10]{SaksS10}
Michael~E. Saks and C.~Seshadhri.
\newblock Estimating the longest increasing sequence in polylogarithmic time.
\newblock In {\em 51th Annual {IEEE} Symposium on Foundations of Computer
  Science, {FOCS} 2010, October 23-26, 2010, Las Vegas, Nevada, {USA}}, pages
  458--467, 2010.

\end{thebibliography}
\end{document}